\documentclass[11pt]{article}

\usepackage{amsfonts,amsmath,latexsym,enumitem,amsthm,amsbsy}
\usepackage{soul}
\usepackage{mathtools}
\usepackage{footmisc}
\usepackage{graphicx}
\usepackage{xcolor}
\usepackage{float}
\usepackage{braket}

\usepackage{tikz}
\usetikzlibrary{fit,shapes, arrows.meta,bending}
\usetikzlibrary{decorations.pathreplacing}

\usepackage[ruled,vlined,linesnumbered]{algorithm2e}
\SetAlFnt{\small}

\usepackage{hyperref}
\hypersetup{
    linktocpage,
    colorlinks,
    citecolor=blue,
    filecolor=blue,
    linkcolor=blue,
    urlcolor=blue
}

\newcommand{\BB}{\vspace*{-\medskipamount}}

\newcommand{\cA}{\mathcal{A}}
\newcommand{\cB}{\mathcal{B}}

\newcommand{\cE}{\mathcal{E}}

\newcommand{\cH}{\mathcal{H}}

\newcommand{\cO}{O}
\newcommand{\cP}{\mathcal{P}}

\newcommand{\cR}{\mathcal{R}}

\newcommand{\logO}{\tilde{\cO}}

\newcommand{\remove}[1]{}

\newcommand{\dk}[1]{#1}
\newcommand{\jo}[1]{{\color{blue}{#1}}}

\newcommand{\Paragraph}[1]{\BB\paragraph{#1}}

\newcommand{\polylog}{\text{polylog }}

\newlength{\pagewidth}
\newlength{\figurewidth}
\makeatletter
\newtheorem*{rep@theorem}{\rep@title}
\newcommand{\newreptheorem}[2]{%
\newenvironment{rep#1}[1]{%
 \def\rep@title{#2 \ref{##1}}%
 \begin{rep@theorem}}%
 {\end{rep@theorem}}}
\makeatother

\newtheorem{theorem}{Theorem}
\newtheorem*{theorem*}{Theorem}
\newreptheorem{theorem}{Theorem}
\newtheorem{lemma}{Lemma}
\newtheorem*{lemma*}{Lemma}

\newtheorem{corollary}{Corollary}

\newtheorem{definition}{Definition}

\newcommand{\ceil}[1]{\left \lceil #1 \right \rceil }
\newcommand{\floor}[1]{\left \lfloor #1 \right \rfloor }
\newcommand{\E}{\mathbb{E}}
\def\Prob{\mathbb{P}\mathrm{r}}

\begin{document}
\title{%
Nearly-Optimal 
Consensus Tolerating Adaptive Omissions: \\
Why is a Lot of Randomness Needed? 
}
\author{
Mohammad T. Hajiaghayi\\
University of Maryland, Maryland, USA.
\and
Dariusz R. Kowalski \\
School of Computer and Cyber Sciences,
Augusta University, Georgia, USA. 
\and
Jan Olkowski\\
\vspace*{-2ex}
University of Maryland, Maryland, USA. \\
}

\date{}





\maketitle

\begin{abstract}
We study the complexity of the problem of reaching agreement in a synchronous distributed system, also called consensus, by $n$ autonomous parties, when the communication links from/to faulty 
parties can omit messages. 
The faulty parties are selected and controlled by an adaptive, full-information, computationally unbounded adversary. 
We design a randomized algorithm that works in 
$O\left(\sqrt{n} \log^{2} n\right)$ rounds and sends 
$O\left(n^2 \log^{3}n\right)$
total number of communication bits, where the number of faulty parties can be $\Theta(n)$. When the number of faulty parties is linear in $n$, our result is simultaneously tight for both these measures within polylogarithmic factors: due to the $\Omega(n^2)$ lower bound on the number of messages send by any Monte Carlo solution, by Abraham et al. (PODC'19), and due to the $\Omega(\sqrt{n / \log{n}})$ lower bound on the number of rounds of any Las Vegas solution by Bar-Joseph and Ben-Or (PODC'98). Thereby, this work settles the landscape of the consensus problem in the omission failures model, which stood as an open question since the work of Dolev and Strong (SICOMP'83).

Additionally, we strictly quantify how much randomness is necessary and sufficient to reduce time complexity to a certain value, while keeping the communication complexity optimal wrt to polylogarithmic factors. We prove that no Monte Carlo algorithm can work in less than $\Omega\left(\frac{n^{2}}{ \max\{R, n\} \log{n}}\right)$ rounds if it uses less than $O(R)$ calls to a random source, assuming a constant fraction of all parties is faulty. This result should be contrasted with a long line of work on consensus algorithms against an {\em adversary limited to polynomial computation time}, thus unable to break cryptographic primitives, culminating in a work by Ghinea et al. (EUROCRYPT'22), where an optimal 
$O(r)$-round solution reaching
consensus with probability $1 - (c r)^{-r}$ is given. Our lower bound strictly separates these two regimes, by excluding such results if the adversary is computationally~unbounded.

On the upper bound side, we show that for $R \in \tilde{\mathcal{O}}\left(n^{3/2}\right)$ there exists a randomized algorithm solving consensus in $\tilde{\mathcal{O}}\left(\frac{n^2}{R}\right)$ rounds, with probability polynomially close to $1$ (whp), 
where tilde notation hides a poly-logarithmic factor. The communication complexity of the algorithm does not depend on the amount of randomness $R$ and stays (universally) optimal within polylogarithmic factors. As a consequence, we give a spectrum of solutions that interpolates between optimal results in the deterministic regime ($R \in O(n)$; $O(1)$ entropy per party) and the randomized regime ($R \in O(n^{3/2})$ random bits).
\end{abstract}

\section{Introduction}


In any distributed system, reaching agreement (consensus) is essential for coordinating actions of
the participating $n$ parties (also called processes), out of which up to $t$ could be faulty. The hardness of the task primarily depends on the type of fault present in the system. The classical hierarchies~\cite{Attiya-Welch-book2004, Lynch-book96} for synchronous message-passing models distinguish the following types of faults, in the order of increasing hardness: crash failures, omission failures, authenticated Byzantine failures, and Byzantine failures. In this work, we focus on the second type of failures in this list, the \textit{omission failures}. 
Precisely, we assume that a computationally unbounded and malicious adversary can observe the system during the computation and \textit{omit} an arbitrary subset of messages send to / received from selected faulty processes in an online, adaptive fashion. The adversary can also, based on the history of the 
computation, corrupt new processes if the number of corrupted stays within a fixed limit $t$.
The adversary, however, cannot see the future random bits.

Despite of huge volume of research on the performance of consensus algorithms, the proper assessment of the hardness of consensus under this type of failure has been elusive. There is no theoretical evidence that omissions failures are weaker than the next model in the hierarchy, authenticated Byzantine, even if such a hypothesis seems compelling. Compared to the weaker model of crash failures, there is a strict hardness barrier following from the existence of an algorithm using $O(n^{3/2} \log^{13/2}{n})$ messages by Hajiaghayi 
et al.~\cite{DBLP:conf/stoc/HajiaghayiKO22} (STOC'22) in the case of crash failures and the $\Omega(n^2)$ lower bound on the number of messages in the model with omission failures, proved by Dolev and Reishuk~\cite{DolevR85} (JACM'85) for deterministic solutions and by Abraham et al.~\cite{AbrahamCDNPRS19} (PODC'19) for the randomized ones, all results assuming $\Theta(n)$ faulty parties.
However, if the space of solutions is categorized based on the round complexity of a solution, even for these two models (i.e., crashes and omissions) the picture is not clear. Both models admit an $\Omega(\sqrt{n/ \log{n}} )$ lower bound for Las~Vegas solutions, due to the work of Bar-Joseph and Ben-Or~\cite{Bar-JosephB98} (PODC'98) (when the number of faulty parties is linear in the system size). In the case of crashes, the lower bound is matched by an algorithm of the same authors~\cite{Bar-JosephB98}, while in the case of omissions -- the best previous solution is 40 years old result of Dolev and Strong~\cite{DolevS83} (SICOMP'83) that works in $O(n)$ rounds! Hence, our first 
goal is to fully understand the hardness of omission failures in 
reaching~consensus:\\
\noindent
{\bf Question 1:} Is there a consensus algorithm, possibly randomized, that at the same time matches both lower bounds with respect to poly-logarithmic factor? That is, is there an algorithm that solves consensus in $\tilde{O}(\sqrt{n})$ rounds with $\tilde{O}(n^2)$ total number of sent communication bits even if a linear number of omission failures occur in the network?

We answer this question affirmatively and give a new algorithm 
that is almost-optimal\footnote{%
In this work, almost-optimal means optimal within a poly-logarithmic factor.} with respect to, simultaneously, the number of rounds {\bf\em and} the total number of sent communication bits when the number of faults is $\Theta(n)$. 

Then, we focus 
on the aspect of how much randomness is necessary to break the $\Omega(n)$ lower bound on the number of rounds for deterministic solutions~\cite{FischerL82} (again, assuming~$\Theta(n)$~faults): \\
\noindent
{\bf Question 2:} What is the impact of the amount of randomness available at processes on the efficiency of consensus algorithms? In particular, could pseudo-random generators be used~efficiently?

We quantify the number of random bits that is sufficient and necessary to achieve a given time $\in \tilde{\Omega}(\sqrt{n})$\footnote{We use tilde notation to hide a polylogarithmic~factor.}, while keeping almost-optimal communication complexity. In particular, we prove that using pseudo-random generators with small seeds may delay reaching consensus nearly quadratically, which may have severe consequences in distributed ledger implementations and distributed database applications based on consensus. Even more interestingly, our lower bounds apply to Monte Carlo algorithms. This makes a surprising distinction between a parallel line of work on consensus protocols in the case of authenticated Byzantine failures governed by a \textit{computationally-bounded} adversary~\cite{abraham2019synchronous, fitzi2003efficient, ghinea2022round, katz2006expected}, done by the cryptography community. In this model, the adversary is still adaptive, but the algorithm has a trusted setup for unique threshold signatures and an unforgeable public-key infrastructure; thus, the view of the adversary is limited to the history of the faulty processes and all messages send to them. 
Nevertheless, the adversary can adaptively corrupt new parties based on its view and has the priority over messages being delivered in the round it corrupts. In this setting, the result of Ghinea et al.~\cite{ghinea2022round} (EUROCRYPT'22) shows an algorithm that terminates in $r$ rounds with probability at least $1 - cr^{-r}$, for any $r$ and some absolute constant $c > 0$. Our lower bound excludes such results in the model with a {\em computationally unbounded adversary}, showing that not only many random bits but also strong cryptography could be necessary for some applications of consensus protocols.

\vspace*{-2ex}
\paragraph{Summary of results and the paper structure.}
As mentioned earlier, we focus on the classical message-passing model in which processes are autonomous, synchronous and can exchange messages via point-to-point communication channels. The adversary corrupts and controls at~most~$t$ processes, causes omission failures at them, and is (full-information) strongly adaptive.
In Sections~\ref{sec:overview-alg},~\ref{sec:overview-lower} and~\ref{sec:overview-randomness}, we present an extended overview of the three main results and novel techniques obtained in this work; for the sake of overview, we typically assume $t=\Theta(n)$. A summary of them can be found in Table~\ref{tab:our_results}. In Section~\ref{sec:related} we discuss the related work in more details. Section~\ref{sec:model} presents detail formal description of the model. 
Sections~\ref{sec:main-algo},~\ref{sec:proof-lower} and~\ref{sec:tradeoff-upper} contain the full and formal analysis of our main results. Section~\ref{sec:future} states major research directions opened by our work.

\begin{table}
\centering
\begin{tabular}{|c|c|c|c|c|c|}
     \hline
     &
     result & time & comm. bits & random bits & comments \\
     \hline
     \hline
     algo- 
     & 
     \underline{Thm~\ref{thm:omissions-opt-res}}
     & $O\left(\sqrt{n}\log^2{n} \right)$ & $O(n^{2} \log^{3}n)$ 
     & 
     $O\left(n^{3/2}\log^2{n} \right)$ & 
     \\
%

     rithms 
     &
    \underline{Thm~\ref{thm:trade-off-res}}
    & $O\left(\frac{n^{2}}{R}\log^2{n}\right)$ & $O\left(n^2\log^5{n}\right)$ & $O\left(R\log^2{n}\right)$ & $\forall R \in O\left(n^{3/2}\right)$
    \\
     \hline
     lower & \cite{Bar-JosephB98} & 
     $\Omega\left(\frac{t}{\sqrt{n\log n}} \right)$ 
     & - & - & correct prob. $= 1$\\
     bounds & \cite{AbrahamCDNPRS19} & - & $\Omega\left(\epsilon t^2\right)$ & - & correct prob. $\ge \frac{3}{4} + \epsilon$\\
      & 
     \underline{Thm~\ref{thm:lower-randomness-res}}
     & $T$ & - & $R$ & $T\times (R+T) = \Omega\left(\frac{t^2}{\log{n}}\right)$ \\
     & 
     
     &  &  &  & correct prob. $\ge 1 - \frac{1}{n^{3/2}}$ \\
     \hline
\end{tabular}
\caption{Our main results presented for three metrics: time complexity, total number of used communication bits and total number of used random bits. All our algorithms are subject to their complexity bounds whp, i.e., with probability polynomially close to $1$, see Section~\ref{sec:model}.
$T$ and $R$ are random variables denoting, respectively, the number of rounds and the cumulative number of random bits used by all processes in a run of an algorithm. The ``correct prob.'' denotes the probability of correctness of the class of algorithms to which the lower bound applies. The new results are underlined in column ``result''. Extended versions of Theorems~\ref{thm:omissions-opt-res} and \ref{thm:trade-off-res}, with explicit dependency on the parameter $t$, are given in Sections~\ref{sec:main-algo} and \ref{sec:tradeoff-upper}, resp.}
\label{tab:our_results}
\end{table}

\section{Model Details and Definitions}
\label{sec:model}

We consider the classical synchronous message-passing distributed system with omission faults, 
cf.,~\cite{Attiya-Welch-book2004,Lynch-book96, raipin2010strongly}.
The system contains $n$ processes, also called parties. Each process has a unique ID in $\cP=[n]=\{1,\ldots,n\}$. For simplicity, we will use $p$ to refer to a process with ID $p \in \cP$. Both $\cP$ and $n$ are known to all process. 
Processes operate in synchronized rounds. Without loss of generality,\footnote{See~\cite{Attiya-Welch-book2004, Bar-JosephB98} for the discussion on generality of our assumptions and related settings -- crash and Byzantine faults.}
we assume that each round consists of the following two phases: 
\begin{enumerate}
    \item \textbf{Local computation phase}: Each process performs a local computation (i.e., autonomously from other processes) in order to change its state. The computation can be any function of the current state, of all messages received prior to this phase, and of a sequence of uniformly distributed and independent random bits of an arbitrary, but finite, length that can be reached by a process at the beginning of the phase. More precisely, for the last 
    parameter of the function, we assume that there exists a random source that, when called, can provide a process and its state-changing function with 
    a 0-1 sequence, of requested length, containing uniform and independent distributed random bits.
    \item \textbf{Communication phase}: During this phase, each process can send messages to any other processes in the system. The content of a message is a function of the current state computed in the preceding local computation phase, and is not limited by the model. In particular, the length and content of messages is not restricted by the model, although our algorithms are designed in a way to use short messages. Each message sent is delivered to its destination at the end of the same phase, and is ready to be processed in the next round, unless an omission failure occurs.
\end{enumerate}
\paragraph{Processes' omission failures and adversaries.}
Not all processes are reliable. Up to some $t = O(n)$ processes may become \textit{(omission)} faulty during the execution -- once a process becomes faulty, it stays faulty through the end of computation and some of its incoming/outgoing messages could be lost. We assume that $t$ is a part of the problem input and thus 
known up-front to
all processes. The decision which process becomes faulty and when, as well as the control over the faulty processes, is governed by an adversary. We consider an adaptive full-power full-information~adversary~that 
\begin{itemize}[leftmargin = 1.5em]
    \item  has unlimited computational power, knows the algorithm and input parameters and can see the states (and thus also the current random bits used) of all processes, as well as 
    the content of all arriving messages, 
    at any time, and
    \item can select online which (non-faulty) processes to fail and when, and with respect to faulty  processes -- it  
    can \textit{omit}
    any subset of messages incoming/outgoing to/from the faulty processes (i.e., such messages are not delivered to their destinations, having the same effect as no~message~sent).
\end{itemize}
Consequently, an {\em adversarial strategy} is a deterministic function, which assigns to each possible history that may occur in any execution some adversarial action for the subsequent \textit{phase} of the execution, i.e., which processes to fail in that phase and in which moment and which messages sent by/to them would reach their destinations. Note, that in the above definition, the adversary has flexibility to adapt its actions between any two phases of the algorithm (not only between consecutive rounds).
In the remainder, we will be referring to this adversary  simply by {\em adaptive~adversary}.

We remark that crash failures of processes can be viewed as omission failures -- the adversary simply bans {\em all} their incoming and outgoing messages after the failure round
(while in the round of a crash, the adversary could allow any subset of outgoing messages to reach their destinations).\footnote{%
In case of crashes, incoming messages are not relevant, because the faulty process could not influence correct processes any more and it is not required from them to satisfy any of the consensus properties.}

\paragraph{Consensus problem.}
A randomized algorithm for processes $\cP = \{1, \ldots, n \}$, 
where each process~$p \in \cP$ holds initial input $b_p \in \{0, 1\}$, is a consensus protocol tolerating $t$ faulty processes if all the three following conditions hold with probability $1$
in the presence of an adaptive adversary failing/controlling at most $t$ processes:
\begin{description}
\item
\textbf{Agreement.} All non-faulty 
processes output the same value.
\item
\textbf{Validity.} If all 
non-faulty processes begin with the same input value $b$, then all of them
output $b$.
\item
\textbf{Termination.} Non-faulty processes have to decide and terminate.
\end{description}
Consider a randomized consensus algorithm against a fixed adversarial strategy. The following metrics determine the quality of the execution against that strategy:
\vspace{-0.05in}
\begin{itemize}
\item \noindent\emph{Time of an execution} of the algorithm is defined as the smallest number $\tau_{1}$ such that the number of rounds that occur by termination of {\em the last non-faulty process} is at most $\tau_{1}$; 
\item \noindent \emph{The number of communication bits in an execution} of the algorithm is the smallest number $\tau_{2}$ such that the total number of bits sent by {\em all processes} in point-to-point messages by termination of the last non-faulty~process is at most $\tau_{2}$;
\item  \noindent\emph{Randomness of an execution} of the algorithm is defined as the smallest number $\tau_{3}$ such that the number of (independent and uniform) random bits used by {\em all processes} by termination of the last non-faulty~process is at most $\tau_{3}$; when describing the lower bound result, we abuse the notation slightly, and define randomness of an execution as the total number of times when processes access their random sources during the local computation phase. Complying to the definition of the local computation phase, in each such access a process can use a sequence of random bits of finite length. Observe that such definition makes a lower bound result even stronger.
\end{itemize}

We define {\em time/communication/randomness complexity of a (distributed) algorithm} 
as a supremum of time, the number of communication bits, and the number of random bits, respectively, taken over all adversarial strategies.\footnote{%
Note that the supremum for different measures could be achieved by sequences of different strategies -- nevertheless, each of the complexities of our solutions is still close to optimal.}
Finally, {\em time/communication/randomness complexity of a distributed problem} is an infimum of all algorithms' time/communication/randomness complexities,~respectively.
In our paper, we are interested in studying almost-optimal solutions, i.e., optimal within a polylogarithmic factor, wrt the abovementioned complexities. The bounds on the complexities should hold with high probability.
We say that a random event occurs \emph{with high probability}
({\em whp} for short), if its probability could be made $1-\cO(n^{-c})$ for any
positive constant $c$ by linear scaling of parameters of the considered random process.\footnote{We would like to note that, after ignoring polylogarithimic factors, all our upper bounds hold with the same asymptotic complexities in even stronger regime where the probabilities are of form $1 - O(n^{\omega(1)})$.}


\section{Main consensus algorithm}
\label{sec:overview-alg}



Our first and main result is a new consensus algorithm which is almost-optimal with respect to time complexity, bit complexity and randomness complexity, if the number of faulty parties is $\Theta(n)$.

\begin{theorem}
\label{thm:omissions-opt-res}
There is a randomized algorithm solving consensus with probability~$1$ against the adaptive omission adversary that can control $t < \frac{n}{30}$ processes, 
which terminates in $O\left(\sqrt{n}\log^{2}{n}\right)$ rounds and uses $O\big(n^2\log^{3}n\big)$ bits of communication and $O\left(n \cdot \sqrt{n}\log^{2}{n}\right)$ random bits, whp.
\end{theorem}

For the ease of presentation, and for the sake of space constraints, in this section, we assume that $t = \frac{n}{30} - 1$ and we provide a more high-level overview of techniques used to obtain the result. A self-contained and fully formal derivation of this theorem, incorporating the upper bound $t$ on the number of faulty processes into the time and communication complexities and containing all the omitted proofs, is deferred to
Appendix~\ref{sec:main-algo}, 
where the above Theorem is restated as Theorem~\ref{thm:main-algo}.

In the case when $t = \Theta(n)$, the almost-optimality of the running time follows from the $\Omega\left(\sqrt{n / \log{n}} \right)$ lower bound showed in~\cite{Bar-JosephB98}. The almost-optimality of the communication bit complexity is due to the result of Abraham et al.~\cite{AbrahamCDNPRS19} who showed that any randomized algorithm solving consensus with at least a constant positive probability against the adaptive omission-causing adversary requires $\Omega(n^2)$ messages (each message carries at least one bit). The almost-optimality of the randomness complexity follows from Theorem~\ref{thm:lower-randomness-res} presented in the later part of the paper. We next give an overview of the algorithm. The pseudocode can be found in Algorithm~\ref{alg:opt-omissions}.

\begin{algorithm}[t!]
\SetAlFnt{\tiny}
\SetAlgoLined
\SetKwInput{Input}{input}
\Input{$\cP$, $p$, $b_p$, t}
$\texttt{operative}_{p} \leftarrow true$, $\texttt{decided}_{p} \leftarrow false$\;\label{line:operative-init}
$V_{p} \leftarrow$ a set of neighbors of $p$ in a predetermined graph $G$ guaranteed by Theorem~\ref{thm:random-graph-properties}\;\label{line:graph_sampling}
$W_{1}, \ldots, W_{\ceil{\sqrt{n}}} \leftarrow $ a pre-defined partition of $\cP$ into 
$\ceil{\sqrt{n}}$ disjoint sets of size $\le \ceil{\sqrt{n}}$ each\;\label{line:sqrt-part}
let $\ell$ be such that $p\in W_\ell$\;\label{line:group}
\For{$\frac{t}{\sqrt{n}}\log n$ epochs\label{line:main-for-begin}}
{
$\texttt{g\_ones}_{p}, \texttt{g\_zeros}_{p}, \texttt{operative}_{p} \leftarrow \textsc{GroupBitsAggregation}(W_{\ell}, p,  \texttt{operative}_{p}; b_{p})$\label{line:group-bits-aggr}\;
\lIf{$\texttt{operative}_{p} = false$}{stay idle until the end of the epoch}

\BlankLine
$\texttt{ones}_{p}, \texttt{zeros}_{p},\texttt{operative}_{p} \leftarrow \textsc{GroupBitsSpreading}(V_{p}, p, \ell, \texttt{operative}_{p}; \texttt{g\_ones}_{p}, \texttt{g\_zeros}_{p})$\label{line:sum_ones_zeros}\;

\BlankLine
\lIf{$\texttt{ones}_{p} > \frac{18}{30}(\texttt{ones}_{p} + \texttt{zeros}_{p})$}{$b_p \leftarrow 1$}\label{line:if-1} 
\lElseIf{$\texttt{ones}_{p} < \frac{15}{30} (\texttt{ones}_{p} + \texttt{zeros}_{p})$}{$b_p \leftarrow 0$}\label{line:if-0}
\lElse{set $b_p$ to $0$ or $1$ uniformly at random}\label{line:if-random}

\BlankLine
\lIf{$\texttt{ones}_{p} > \frac{27}{30}(\texttt{ones}_{p} + \texttt{zeros}_{p})$ or $\texttt{ones}_{p} < \frac{3}{30} (\texttt{ones}_{p} + \texttt{zeros}_{p})$}{$\texttt{decided}_{p} \leftarrow true$}\label{line:if-1-r}

}\label{line:main-for-end}

\lIf{$\texttt{operative}_{p} = true$ and $\texttt{decided}_{p} = true$}{\textbf{send} $b_{p}$ to all processes in $\cP$\label{line:good-spread}}
\lElseIf{any message $b_{q}$ received from some process $q$}{
$b_{p} \leftarrow b_{q}$\label{line:final-receiv-1}}
\tcc{in the above, $q$ can be chosen arbitrarily from the received messages}
\BlankLine
\lIf{$\texttt{decided}_{p} = true$ or $(\texttt{operative}_{p} = false$ and $p$ received a message in the previous round$)$}{\textbf{decide} $b_{p}$\label{line:if-decided}}
\Else{
\lIf{$\texttt{operative}_{p} = true$}{$p$ participates in the deterministic synchronous Consensus algorithm given in Theorem~4 in~\cite{DolevS83} with the input bit $b_{p}$; if $p$ reaches agreement in that protocol, it broadcasts the decision to all processes in $\cP$ and it decides on the algorithm's decision\label{line:spread-2}}
\lElse{$p$ remains idle until a decision is sent to it; upon receiving a decision, it decides on this value\label{line:final-receiv-2}}
}\label{line:fixing-und}
\caption{\textsc{OptimalOmissionsConsensus}}\label{alg:opt-omissions}
\end{algorithm}

\paragraph{Universal idea: Local and dynamic partitioning of processes into operative / inoperative and implementing time- and communication-efficient biased-majority-voting 
only by the operative ones.}
We introduce a new partitioning of processes into operative and inoperative, 
based on the communication received by each process from a certain pre-defined set of other processes which maintain their operative status (this set may vary, depending on what procedure is executed -- it will be emphasized later). 
This partition is not equivalent to the standard classification into faulty / non-faulty ones. 
With our partition,
we can guarantee that faulty processes either communicate well enough to contribute to the progress towards a unified decision (i.e., stay operative) or become excluded from the set of operative processes, having no impact on the final decision. Our partition also avoid a major performance problem in omission-tolerant or Byzantine-tolerant computation -- identifying a single faulty process may require at least quadratic number of messages, cf.,~\cite{AbrahamCDNPRS19}, which makes it fast, local and incorporated in efficient communication~schedules.

Then, we employ the idea of reaching consensus by applying
the biased-majority-voting rule, as proposed in~\cite{Bar-JosephB98}, but with a novel twist -- only the operative processes implement the vote protocol to agree on a consensus decision. It uses $O(\sqrt{n}\log{n})$ repetitions of the single vote subroutine, called {\em epochs} in the pseudo-code, see lines~\ref{line:group-bits-aggr}-\ref{line:if-1-r} in Algorithm~\ref{alg:opt-omissions}. Each repetition/epoch consists of $O\left(\log n\right)$ rounds of communication-efficient counting, see the description below), however it succeeds in unifying the votes only if the number of newly failed processes is $O(\sqrt{n})$ and only with a constant probability (this is why we need $O(\sqrt{n}\log{n})$ epochs).
Only after the part implementing the biased-majority-voting rule ends, the operative processes communicate the decision to the remaining parties. 
We first describe how we implement a single epoch,
based on two technical advancements, and conclude with more details on how the overall consensus protocol (based on the biased-majority-voting rule) is designed.

\Paragraph{An implementation of a single biased-majority-vote subroutine (epoch).}
In our case, a single epoch (i.e., a single repetition of the biased-majority-vote subroutine, lines~\ref{line:group-bits-aggr}-\ref{line:if-1-r}) heavily relies on 
{\bf counting}, collaboratively by every operative process,
the number of operative processes that have candidate decision value $0$ and, separately, value $1$. These numbers must be approximate, up to an additive factor linearly dependent on the number of processes that become inoperative, as the operative status may change dynamically -- some processes can lose it before the calculation finishes and, in consequence, their candidate values might not be properly counted by others. Moreover, this calculation has to consider the fact that some operative processes can be 
controlled by adversary, thus it must, regardless, exploit the property that the operative processes communicate with enough other processes. A protocol performing
this calculation in $O(\log{n})$ rounds and using $O(n^{3/2} \log^2{n})$ communication bits, in total, is the main technical advancement of this algorithm and we present it next in a form of {\bf\em two technical advancements.}

\Paragraph{{\bf\em Technical advancement 1: $\sqrt{n}$-decomposition into groups and binary-tree-like intra-group calculations of operative processes for communication saving.}}
As mentioned above, there are some inherit difficulties in the omission failure model that complicate the time- and communication-efficient counting of the number of candidate values (i.e., votes) $0$ and $1$ among the operative processes. Here, we present the techniques we use to mitigate them. We pre-define fixed partition of the set of processes into $\ceil{\sqrt{n}}$ {\em groups} of size $\floor{\sqrt{n}}$ or $\ceil{\sqrt{n}}$ each (see line~\ref{line:sqrt-part} of Algorithm~\ref{alg:opt-omissions} and example in Figure~\ref{fig:group-partition}), and first require the operative processes to count the number of operative $0$'s and $1$'s only within the groups (executed procedure \textsc{GroupBitsAggregation} in line~\ref{line:group-bits-aggr}). In this part, we use a virtual sparse data structure on some subsets of processes, structured into a balanced binary tree, which is used to aggregate the counts.

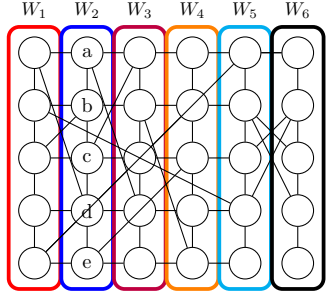
\begin{figure}
\centering
\begin{tikzpicture}[scale=0.7, transform shape]
    \foreach \x in {0,...,5} {
        \foreach \y in {0,...,4} {
            \node[circle, draw, minimum size=0.6cm] (\x\y) at (\x, \y) {};
        }
    }
    \node[draw=red, rounded corners, line width=0.5mm, fit=(00) (04), label={[xshift=0.0cm, yshift=0.0cm, font=\small]$W_{1}$}] {};

    \node[draw=blue, rounded corners, line width=0.5mm, fit=(10) (14), label={[xshift=0.0cm, yshift=0.0cm, font=\small]$W_{2}$}] {};
    \node[] at (14) {a};
    \node[] at (13) {b};
    \node[] at (12) {c};
    \node[] at (11) {d};
    \node[] at (10) {e};
    
    \node[draw=purple, rounded corners, line width=0.5mm, fit=(20) (24), label={[xshift=0.0cm, yshift=0.0cm, font=\small]$W_{3}$}] {};
    
    \node[draw=orange, rounded corners, line width=0.5mm, fit=(30) (34), label={[xshift=0.0cm, yshift=0.0cm, font=\small]$W_{4}$}] {};
    
    \node[draw=cyan, rounded corners, line width=0.5mm, fit=(40) (44), label={[xshift=0.0cm, yshift=0.0cm, font=\small]$W_{5}$}] {};

    \node[draw=black, rounded corners, line width=0.5mm, fit=(50) (54), label={[xshift=0.0cm, yshift=0.0cm, font=\small]$W_{6}$}] {};

    \draw (00) -- (01);
    \draw (00) -- (10);
    \draw (01) -- (11);
    \draw (01) -- (02);
    \draw (02) -- (12);
    \draw (02) -- (03);
    \draw (03) -- (13);
    \draw (03) -- (04);
    \draw (04) -- (14);
    \draw (10) -- (11);
    \draw (10) -- (20);
    \draw (11) -- (12);
    \draw (11) -- (21);
    \draw (12) -- (13);
    \draw (12) -- (22);
    \draw (13) -- (14);
    \draw (13) -- (23);
    \draw (14) -- (24);
    \draw (20) -- (21);
    \draw (20) -- (30);
    \draw (21) -- (22);
    \draw (21) -- (31);
    \draw (22) -- (23);
    \draw (22) -- (32);
    \draw (23) -- (24);
    \draw (23) -- (33);
    \draw (24) -- (34);
    \draw (30) -- (31);
    \draw (30) -- (40);
    \draw (31) -- (32);
    \draw (31) -- (41);
    \draw (32) -- (33);
    \draw (32) -- (42);
    \draw (33) -- (34);
    \draw (33) -- (43);
    \draw (34) -- (44);
    \draw (40) -- (41);
    \draw (41) -- (42);
    \draw (42) -- (43);
    \draw (43) -- (44);
    \draw (00) -- (11);
    \draw (02) -- (13);
    \draw (04) -- (11);
    \draw (14) -- (21);
    \draw (23) -- (30);
    \draw (41) -- (53);
    \draw (43) -- (52);
    \draw (44) -- (54);
    \draw (00) -- (22);
    \draw (11) -- (33);
    \draw (22) -- (44);
    \draw (10) -- (32);
    \draw (03) -- (41);
    \draw (12) -- (24);

    \draw (50) -- (51);
    \draw (51) -- (52);
    \draw (52) -- (53);
    \draw (53) -- (54);
    \draw (51) -- (43);
    \draw (53) -- (42);
    \draw (54) -- (44);

\end{tikzpicture}
\caption{A schematic picture of two different techniques used for communication between processes. Different colors represent different groups in the $\sqrt{n}$-decomposition of the processes. The links represent the overlaying communication resembling a sparse random graph used for exchanging operative counts of different groups. The choice of links is independent of the $\sqrt{n}$-decomposition.\label{fig:group-partition}}
\end{figure}

Vertices of the binary tree correspond to specific subsets of processes in the group: leaves of the binary tree are singletons in the group, and
each vertex in a higher layer corresponds to the union of the subsets that are already identified with the children of that vertex in the tree.
The root of the binary tree corresponds to the set of all processes in the group. Processes calculate the number of operative $0$'s and $1$'s, called {\em operative counts}, starting from 
the leaves of the tree (i.e., singletons) and then keep moving up the tree. 
At each vertex in a higher layer, the processes in
the subset corresponding to that vertex work together to relay and sum up the operative counts from the lower layer. 
In this relay-and-aggregation procedure, all operative processes in the group exchange messages, not only to relay and aggregate values from children to parents in the virtual tree, but also to keep track who remains operative
in the whole group. More precisely,
if a process receives information from less than half of the other processes from the group, during the procedure of 
relaying and aggregating the operative counts of $0$'s and $1$'s from the lower layer, it becomes inoperative.
See example in Figure~\ref{fig:single-group-tree}.
After $O(\log{n})$ rounds, corresponding to the height of the binary tree,  the operative counts for the entire group -- corresponding to the root of the tree -- can be calculated. Complying with the rule of receiving enough number of messages in order to maintain an operative status, only the processes of the group that remain operative use the operative counts further in the protocol. In the analysis, we will be able to prove that, regardless of the omission failures' pattern, there is a group of $\Theta(n)$ operative processes whose counts of operative $0$'s and $1$'s differ by, at most, the number of processes who have become inoperative. Taking the advantage of the binary-tree-structured communication, we can guarantee that processes of each group exchange at most $\logO(n)$ bits in total and that the procedure of operative counting $0$'s and $1$'s takes only $O\left(\log{n}\right)$ rounds. Summarizing, we prove the following result about a single execution of the procedure $\textsc{GroupBitsAggregation}$. Detailed description can be found in Appendix~\ref{subsec:comm-patterns} and in Algorithm~\ref{alg:bits-agreg}; the formal analysis is in Appendix~\ref{subsec:analysis-main}.

\begin{lemma*}[Lemma~\ref{lem:bits-agg-contr} and~\ref{lem:msg-aggr} in Appendix~\ref{subsec:analysis-main}]
A single execution of the procedure \textsc{GroupBitsSpreading} works in $O(\log{n})$ rounds and guarantees that every operative process in a single group knows an approximate number of other operative processes in the group having candidate value $0$ and $1$. 
The numbers in different operative processes differ by at most the number of processes of the group that became inoperative during the execution of the procedure.
Processes in a single group use at most $O(n\log^2{n})$ bits of communication, in total, during the execution.
\end{lemma*}

\begin{figure}
\centering
\begin{tikzpicture}[
    scale=0.7, transform shape,
    level 1/.style={sibling distance=60mm},
    level 2/.style={sibling distance=30mm},
    level 3/.style={sibling distance=15mm},
    ball/.style={draw, ellipse, minimum width=40pt},
    main_ball/.style={draw, circle, minimum width=60pt, line width=0.4mm},
    edge from parent/.style={draw},
    ball_e/.style={white, ellipse, minimum width=40pt},
]
\node[ball] (root) {a,b,c,d,e}
    child {
        node[ball] (abcd) {a,b,c,d}
        child {
            node[ball] {a,b}
            child { node[ball] (a) {a} }
            child { node[ball] (b) {b} }
        }
        child {
            node[ball] {c,d}
            child { node[ball] {c} }
            child { node[ball] {d} }
        }
    }
    child {
        node (x) [ball] {e}
    };
\draw[dotted] (x) -- ++(-1.5,-1.5);
\draw[dotted] (x) -- ++(1.5,-1.5);

\node (main) at (8, -2.0) [main_ball, label={[xshift=-0.8cm, yshift=-2.7cm, font=\large]$\textcolor{pink}{1}$}] {$a,b,c,d,e$}; 

\node[draw=pink, rounded corners, line width=0.3mm, fit=(abcd)(x), label={[xshift=0.5cm, yshift=-2.0cm, font=\large]$\textcolor{pink}{2}$}] (lower) {};

\node[draw=pink, rounded corners, line width=0.3mm, fit=(root), label={[xshift=1.4cm, yshift=-1.2cm, font=\large]$\textcolor{pink}{3}$}] (upper) {};

\draw[draw=pink, fill=pink, {Stealth[length=15pt,width=10pt]}-{Stealth[length=15pt,width=10pt]}] (lower) edge [line width=0.3mm, bend right=45] node[midway,below] {$\{a,b,d,e\}$}  (main);

\draw[draw=pink, fill=pink, -{Stealth[length=15pt,width=10pt]}] (main) edge [line width=0.3mm, bend right=15] node[midway,below] {$\{a,b,d,e\}$}  (root);

\node[inner sep=6pt, draw=blue, rounded corners, line width=0.5mm, fit=(x) (a)(b)(root)(main), ] {};
\end{tikzpicture}
\caption{Visualization of the $\sqrt{n}$-decomposition of the blue group from Figure~\ref{fig:group-partition}. The processes $a,b,c,d,e$ in the group are logically decomposed into a binary tree. The pink arrows visualize the three-round process of relaying operative counts of the two children of the root to the root itself. First, the counts are relayed to all processes in the group (arrow \#1), then the processes send a confirmation if they received the counts (arrow \#2), finally, all in the group transmit the received counts to the higher layer -- the root in this case (arrow \#3). 
Some processes can be faulty (process $c$ does not communicate, only $\{a,b,d,e\}$)
and their 
values are not guaranteed to be accumulated~accurately.
\label{fig:single-group-tree}}
\end{figure}
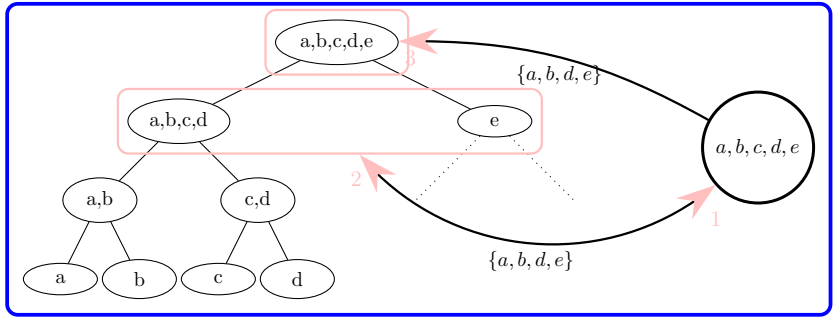


\Paragraph{{\bf\em Technical advancement 2: Fast inter-group communication and status maintenance between operative processes.}} After the tree-based communication, the operative processes in each group have a shared knowledge about the count of operative $0$'s and $1$'s within their group. Since there are $\ceil{\sqrt{n}}$ different groups, the number of logically different counts is $O(\sqrt{n})$. 
    To exchange these $O(\sqrt{n})$ counts between the groups, the operative processes communicate along the links in a sparse, but well-connected, graph -- neighborhoods of which are pre-selected locally in line~\ref{line:graph_sampling} -- that underlays the entire network; see the executed procedure \textsc{GroupBitsSpreading} in line~\ref{line:sum_ones_zeros} of Algorithm~\ref{alg:opt-omissions} and the illustration of the graph on the top of the group partition in Figure~\ref{fig:group-partition}. The graph used by the operative processes is selected as follows. 
    \dk{We consider a random graph, where each edge is selected independently at random with probability $\Theta(\log{n} / n)$. Next, we use the probabilistic analysis to show that the following event holds whp: a random graph with such edge density has the property that every subgraph of a constant-fraction size is dense and shallow -- see Theorem~\ref{thm:random-graph-properties} in Appendix~\ref{subsec:analysis-main}. The ``dense'' property refers to the fact that removing an arbitrary but at most $\alpha$-fraction of edges incident to any vertex from a linear number of vertices, allows to find a connected subgraph of this linear number of vertices such that every vertex has degree at least $\beta \log{n}$ within this subgraph. The ``shallow'' property refers to the fact that the latter subgraph has logarithmic diameter (asymptotically).} These two properties justify our partition into operative and inoperative classes, from perspective of the inter-group communication: as long as the process has more than $\beta \log{n}$ active links, intuitively it belongs to the connected shallow subgraph and thus it is capable of exchanging information with any other process with this property in $O(\log{n})$ rounds. This 
holds regardless of the factual faulty / non-faulty state of the processes. Therefore, the operative processes can spread among themselves the operative counts of the $\ceil{\sqrt{n}}$ different groups, yielding the property that as long as a process remains operative it knows {\em some operative counts of any group} 
with at least one operative process. More specifically, every operative process stores a data structure memorizing the operative counts of each of the $\ceil{\sqrt{n}}$ groups present in the system. Initially, it knows only the counts of the group it belongs to. 
In $\Theta(\log{n})$ rounds of communication it keeps sending these counts along every edge determined by the underlying graph, maintaining the fact that operative counts of a particular group are sent only once via each edge; at the same time it receives counts of other groups and updates the data structure based on this information. In case a process receives two or more different count values of some group, it can choose arbitrarily any of them -- as argued earlier, all of them could differ by at most the number of processes that have become inoperative in that group. To summarize, the procedure ~\textsc{GroupBitsSpreading} has the following outcome (the formal description is given in Algorithm~\ref{alg:bits-spread} and in Appendix~\ref{subsec:comm-patterns}; the formal analysis is provided in Appendix~\ref{subsec:analysis-main}).

\begin{lemma*}[Lemmas~\ref{lem:spreading-reaching},~\ref{lem:operative-contribution} and Theorem~\ref{thm:main-algo} in Appendix~\ref{subsec:analysis-main}]
Assume that processes run the procedure~\textsc{GroupBitSpreading} with $O(\sqrt{n})$ different logical input values (which are the operative counts of candidate values $0$ and $1$ of each group). At the end of the procedure, each operative process knows at least one copy of the logical value, provided that at least one process starting with this logical value remains operative. The procedure uses $O(\log{n})$ communication rounds and $O(n\sqrt{n}\log^{2}{n})$ communication bits in total.
\end{lemma*}

The combination of the two above technical advancements lead to calculating the number of candidate values 0 and 1 among the operative processes. Although these numbers are approximate, as they can differ by the number of processes that have become inoperative, this difference is acceptable to still employ a variant of the biased-majority-voting consensus,
as we discuss in the next part.
\Paragraph{Putting them all together: consensus protocol based on biased-majority-voting adjusted to the new efficient voting implementation in an epoch.} 
Assuming that each operative process has the approximate numbers of other operative processes having a candidate value $1$ and $0$, we explain the modification to the consensus framework based on the biased-majority-voting by Bar-Joseph and Ben-Or~\cite{Bar-JosephB98} in order to adapt to the dynamic characteristic of the operative set of processes and to the properties that our new efficient single-epoch implementation of a voting has. Our modification takes into account that, in our case, the counts of operative $0$'s and $1$'s are not the same in every operative process. The communication protocols guarantee only that a candidate value of an operative process is accounted by any other operative process, however, it gives no guarantee regarding the candidate values of the processes that become inoperative during the calculation and communication. Also, in our implementation of a single biased-majority-voting subroutine (epoch), 
described earlier, the operative processes do not assign any default values to the candidate values of the inoperative processes. 
Thus, any operative process estimates the number of all operative processes simply by adding the operative counts of $0$'s and $1$'s. 
Based on the above estimates, the algorithm employs the 
procedure of converging the candidate values of the operative processes to the final decision value according to the following: if an operative process has an operative count of $1$'s (meaning the candidate values of operative processes assigned to $1$) at least $\frac{18}{30}$ of the estimated total of all operative processes, it sets the candidate value to 1. If the operative count is less than half of the estimated total, it sets the candidate value to 0. In all other cases, the candidate value for the next step is a uniformly chosen random bit. See lines~\ref{line:if-1}-\ref{line:if-random} and Figure~\ref{fig:random-coin} for an illustration.
\begin{figure}[!h]
\centering
\begin{tikzpicture}[scale=0.7, transform shape]
  \draw (0,0) -- (15,0);

  \foreach \x/\text in {0/0, 2/$0 + \frac{t}{n}$, 7.5/$\frac{1}{2}$, 9.5/$\frac{1}{2} + \frac{t}{n}$, 13/$1 - \frac{t}{n}$, 15/$1$}
    \draw[-](\x, -2pt)--(\x, 2pt);
    
  \foreach \x/\text in {0/0, 2/$0 + \frac{3}{30}$, 7.5/$\frac{1}{2}$, 9.5/$\frac{1}{2} + \frac{3}{30}$, 13/$1 - \frac{3}{30}$, 15/$1$}
    \fill (\x,0) circle (0pt) node[below, yshift=-5pt] {\text};

  \draw[decorate,decoration=brace](0 cm + 7pt, 15pt) -- node[above,font=\small]{$b = 0$} (7.5 cm - 7pt, 15pt);

  \draw[decorate,decoration=brace](7.5 cm + 7pt, 25pt) -- node[above,font=\small]{$b ~\sim \mathcal{B}\left(1, \frac{1}{2}\right)$} (9.5 cm - 7pt, 25pt);
  
  \draw[decorate,decoration=brace](9.5cm + 7pt, 15pt) -- node[above,font=\small]{$b = 1$} (15cm - 7pt, 15pt);

  \draw[blue, fill=blue] (8.75,5pt) circle (2pt);
  \draw[blue, fill=blue]  (9,5pt) circle (2pt);
  \draw[blue, fill=blue] (10,5pt) circle (2pt);
  \draw[blue, fill=blue] (11,5pt) circle (2pt);
  
  \draw[green, fill=green] (1.75, 5pt) circle (2pt);
  \draw[green, fill=green] (2.5, 5pt) circle (2pt);
  \draw[green, fill=green] (3, 5pt) circle (2pt);
  \draw[green, fill=green] (3.23, 5pt) circle (2pt);

  \draw[gray, fill=gray] (1.35, 5pt) circle (2pt);
  \draw[gray, fill=gray] (13.65, 5pt) circle (2pt);
  \draw[gray, fill=gray] (13.25, 5pt) circle (2pt);

  \node[below, align=center] at (7.5, -1.0) {
    $\{$ \tikz{\draw[green, fill=green] circle (2pt);} ,  \tikz{\draw[blue, fill=blue] circle (2pt);} $\}$ - valid distributions of candidate values $b$ across different processes \\
    
    $\{$ \tikz{\draw[gray, fill=gray] circle (2pt);} $\}$ - invalid distribution of the same values; ratios at different processes differ by more than $\frac{3t}{n}$
  };
  
\end{tikzpicture}
\caption{A picture explaining the thresholds in a single execution of the biased-majority-voting subroutine, see lines~\ref{line:if-1}-\ref{line:if-random} in Algorithm~\ref{alg:opt-omissions}. Different colors represent different outcomes \dk{(each obtained in a different epoch)} of the counting of candidate values in preceding lines~\ref{line:group-bits-aggr} and~\ref{line:sum_ones_zeros}. \label{fig:random-coin}}
\end{figure}
Following the main line of the biased-majority-voting idea, in the analysis, we show that if a perturbation in the number of operative processes before the counting starts and after it ends is small, i.e., $O(\sqrt{n})$, the estimations are good enough to reach a consensus decision by operative processes, with constant probability i.e. the candidates values among operative processes are unified.
A high level intuition is that if there is a fraction of $1$'s as the candidate values among the operative processes of at least $\frac{18}{30}$ then every operative processes assigns $1$ as its candidate value. Otherwise every operative process either sets a random bit or $0$ as the candidate value, however, the standard bounds on the deviation from the mean of the sum of many i.i.d. random variables guarantee that, with constant probability, the number of assigned $0$'s is below the mean by $\Theta(\sqrt{n})$. If the perturbation is indeed $O(\sqrt{n})$, which occurs with a constant probability, this leads to assigning $0$ as the candidate value in the next repetition of the biased-majority-voting subroutine by every operative process. Since the convergence analysis of a similar procedure has been done earlier~in~\cite{Bar-JosephB98}, but with different constants, we omit the details here and refer the reader to Lemmas~\ref{lem:three-epochs},~\ref{lem:deciding-1} in~Appendix~\ref{subsec:analysis-main}.

On the other hand, the communication graphs in the communication protocols are designed to guarantee at least $n - \Theta(t)$ operative processes in the system, regardless of the adversary's actions.
\begin{lemma*}[Lemma~\ref{lem:good-proc-are-large} in Appendix~\ref{subsec:analysis-main}]
The number of operative process is always at least $n - 3t$.
\end{lemma*}
Therefore, after $O\left(\frac{t}{\sqrt{n}}\log{n}\right) \le O(\sqrt{n}\log n)$ epochs, each executing our biased-majority-voting subroutine, we can assert, by a counting argument, that there were at least $\Omega(\log{n})$ epochs in which the perturbation to the number of operative processes was $O(\sqrt{n})$. In consequence, we can argue that operative processes have reached a consensus decision with high probability -- the voting in different epochs use independent random bits, thus their outcome are independent and standard bounds on probability of success of $\Theta(\log{n})$ independent trials of a Bernoulli variable can be applied.
As a final step, the operative processes disseminate the decision to all processes by all-to-all-communication, see lines~\ref{line:good-spread}-\ref{line:if-decided}.\footnote{%
To increase the probability of success to $1$, after each epoch, the operative processes employ a safety rule: if the estimation of operative processes holding candidate value $1$'s constitutes at least $\frac{27}{30}$ fraction of the overall estimation of all operative processes (or, conversely, the estimation of $1$'s constitutes less than $\frac{3}{30}$ fraction), a process sets an auxiliary variable $\texttt{decided}$ to $true$, see line~\ref{line:if-1-r}. The undecided processes may eventually switch to a deterministic protocol (line~\ref{line:spread-2}), working in $O(n)$ rounds and sending $O(n^3)$ communication bits cf.~\cite{DolevS83}, but the underlying idea is that, based on the previous arguments, it is only with a probability of less than $O\left(\frac{1}{n}\right)$ that an operative process remains undecided throughout the $O\left(\frac{t}{\sqrt{n}}\log{n}\right) \le O(\sqrt{n}\log n)$ repetitions of our biased-majority-voting subroutine. As this part is highly technical, we postpone details to Section~\ref{sec:main-algo}.} 

The detailed analysis of Algorithm~\ref{alg:opt-omissions} is presented in Appendix~\ref{subsec:analysis-main} and the precise formal result is stated in Theorem~\ref{thm:main-algo}, however, we note here the complexity milestones. The communication bit complexity is always $O\left(\frac{t}{\sqrt{n}}\log{n} \cdot n^{3/2}\log^2{n} + n^2 + n t \right) = O\left(n\left(t \log^{3}n + n\right)\right)$. The first additive term corresponds to $O\left( \frac{t}{\sqrt{n}}\log{n}\right)$ repetitions of our biased-majority-voting subroutine
(i.e., epochs), each requiring $O\left(n^{3/2}\log^2{n} \right)$ bits of communication. The additive term $O\left(n^2\right)$ corresponds to the procedure of informing other processes of the decision made by the 
operative processes. The last term $O(nt)$ corresponds to the execution of the deterministic protocol. The number of random bits used is at most $1$ per process per the repetition of our biased-majority-voting subroutine (i.e., once per epoch), which gives $O\left(t \sqrt{n} \log{n}\right) \le O\left(n \sqrt{n} \log{n}\right)$ bits of randomness in total. Similarly, each repetition of an 
the counting scheme takes 
$O(\log {n})$ rounds, 
which implies $O\left(\frac{t}{\sqrt{n}} \log^2{n}\right)\le O\left(\sqrt{n} \log^2{n}\right)$ round complexity, conditioned on the fact that the deterministic consensus algorithm is not evoked -- this happens whp. In the case where the deterministic algorithm is executed, it takes additional $O(t)$ rounds after which all non-faulty processes decide with probability~$1$, but this happens with polynomially small probability.

\section{Lower bound}
\label{sec:overview-lower}

Our next result shows a new connection between randomness, which we informally define as the total number of accesses to a random source by the algorithm,\footnote{Observe that this definition is, in principle, more general than counting only the total number of random bits. 
See Section~\ref{sec:model} for formal specification of randomness.} and time complexity in consensus solutions against an adaptive adversary. It provides a lower bound for a broader class of consensus algorithms that are correct with high probability,\footnote{%
Algorithms correct with probability $1$ are also in this class.}
against an adaptive adversary who can crash processes permanently -- it automatically extends to omissions and more severe faults. Compared to previous lower bounds, we introduce a new, amortized analysis of the valency framework~\cite{Bar-JosephB98, FischerLP85, moses1998unified} -- a common tool to deriving lower bounds for consensus algorithms. By crafting a new parameterized approach to the coin-flipping game -- an abstraction that models processes' random choices -- we can fully adapt the power of adversary to the amount of randomness the algorithm is using.

\begin{theorem}\label{thm:lower-randomness-res}
For a synchronous algorithm solving consensus with probability $\ge 1 - \frac{1}{n^{3/2}}$, let $T$ denote the number of rounds in an execution of the algorithm and $R$ be the total number of times the processes have accessed a random source. 
There exists an adaptive adversarial strategy that~guarantees, with probability at least $1 - \frac{1}{\log{n}}$,\footnote{In fact, for the cost of 
poly-logarithmic times more faults, the bound could be polynomialy~close~to~one.} 
\vspace*{-2ex}
\[
T \times \left(R+T\right) = \Omega\left(\frac{t^{2}}{\log{n}}\right)
\ .
\]
\end{theorem}

The above result is proved in
Appendix~\ref{sec:proof-lower}, as Theorem~\ref{thm:lower-randomness}.
It shows a \textit{trade-off} between fast algorithms and algorithms that are frugal in their calls to a random source. 
Lower bounds of similar flavor have been already known for the asynchronous setting. Aspnes (JACM'98)~\cite{Aspnes98} was the first to show that $\Omega\left(\frac{t^{2}}{\log^2{t}}\right)$ coin flips are needed to solve asynchronous consensus. That result has been later improved by Attiya and Censor (STOC'07)~\cite{DBLP:conf/stoc/AttiyaC07}, who obtained a tight bound showing that asynchronous consensus requires $\Theta\left(n^{2}\right)$ (asynchronous) step complexity. 
Note that the adversary in an asynchronous setting is much more powerful than in the synchronous one, considered in this paper, since he can delay {\em all} operations arbitrarily -- we obtain similar lower bound (right-hand side of our formula) without being able to use an advantage of asynchrony in adversarial strategies.

In case of synchronous algorithms, the only known  lower bound was obtained by Bar-Joseph and Ben-Or~\cite{Bar-JosephB98}, who showed that no algorithm can solve consensus with probability $1$ against the adaptive adversary in fewer than $\Omega(\frac{t}{\sqrt{n\log{n}}})$ rounds, even with unlimited randomness. 
Since a single process can have at most one call to a random source in a round, substituting $R := n \times T$ in our lower bound also implies that any algorithm solving consensus with probability at least $1 - \frac{1}{n^{3 / 2}}$ has to work for $T$ rounds such that $T \times T \times (n + 1) = \Omega(\frac{t^{2}}{\log{n}}) \implies T = \Omega(\frac{t}{\sqrt{n\log{n}}})$  with probability at least $1 - \frac{1}{\log{n}}$. Thus, we extend the bound in~\cite{Bar-JosephB98} to a wider class of algorithms -- Monte Carlo solutions.
On the other hand, in the case of having a small (e.g., constant or polylogarithmic) number of processes who make random calls in a round, or when process uses pseudo-random generators of small (e.g., polylogarithmic) seed and the number of failures is big (e.g., $t=\Omega(\frac{n}{\polylog n})$) -- the lower bound implies $T = \Omega(\frac{t}{\polylog n})$, which is (almost) equivalent to the renowned lower-bound $T = \Omega(t)$ (cf.~\cite{Attiya-Ellen-book-2014}) for 
deterministic executions. Nevertheless, our lower bound describes precisely the relationship between the time and the number of random calls in the entire spectrum between the two extremes of unlimited randomness and determinism.

\paragraph{Lower bound's technical novelty and overview of its analysis.}
We propose a new improved analysis of the one-round coin-flipping game -- an abstraction proposed by Bar-Joseph and Ben-Or~\cite{Bar-JosephB98}. They showed that, from a high level perspective, if $n$ processes use randomness in a round, the adversary (knowing the random outcomes) can hide $\Theta(\sqrt{n\log{n}})$ values (which corresponds to failing $\Theta(\sqrt{n\log{n}})$ processes) such that with probability at least $1 - \frac{1}{n}$ the execution cannot be close to deciding. 
We improve this analysis and make it {\em parameterized} with respect to the number of calls to a random source, see Lemma~\ref{cor:coin-game} in Appendix~\ref{sec:proof-lower}, in order to be able to {\em amortize} this number in the final analysis in the~proof~of~Theorem~\ref{thm:lower-randomness-res} (equivalently, Theorem~\ref{thm:lower-randomness} in Appendix~\ref{sec:proof-lower}).

More specifically, using Talagrand's concentration inequality, we show that even if only $k < n$ processes decide to make a call to a random source, the adversary can fail at most $\Theta(\sqrt{k\log{\alpha}})$ of them in such a way that the probability of preventing decision is at least $1 - 2^{\alpha}$, for $\alpha < \frac{1}{2}$. Introducing the artificial parameter $\alpha$ allows the adversary to control the game with almost any desired precision with the cost of only $\log{\alpha}$ times more failures in an execution.  

Having this tool in 
hand, 
one could follow the generic framework of analyzing valency of the executions, as used by Bar-Joseph and Ben-Or in~\cite{Bar-JosephB98} and also in other related contexts 
in distributed computing, 
c.f.,~\cite{Aspnes98},~\cite{DBLP:conf/stoc/AttiyaC07}. In short, the framework relies on partitioning the executions into a finite and small number of exclusive types (also referred to as valency types) of executions that capture the probabilities of the algorithm deciding $0$ or $1$, given the history of the execution up to some point in time. 
More precisely, let $\Prob(\cH,\cA)$ be the probability of reaching consensus on value $1$ when continuing the run of the algorithm with history $\cH$ under adversarial strategy 
    $\cA$.
We say that a 
state of the algorithm in round $i$, defined uniquely by its history $\cH$, 
is: 

\begin{itemize}
    \item 
{\em null-valent} if for all adversarial strategies 
$\cA$ 
extending this state, 
we have $\frac{1}{n\log n} - \frac{i}{n^{2}} \le \Prob(\cH,\cA)\le 1-\frac{1}{n\log n} + \frac{i}{n^{2}}$,
\item
{\em $1$-valent} if there is an adversarial strategy 
$\cA$ 
extending this state
such that $\Prob(\cH,\cA)> 1-\frac{1}{n\log n} + \frac{i}{n^{2}}$ and for every other adversarial strategy 
$\cA'$: 
$ \Prob(\cH,\cA') \ge \frac{1}{n\log n} - \frac{i}{n^{2}}$,
\item
{\em $0$-valent} if there is an adversarial strategy $\cA$ extending this state, such that $\Prob(\cH,\cA) < \frac{1}{n\log n} - \frac{i}{n^2}$ and for every other adversarial strategy $\cA'$: $\Prob(\cH,\cA') \le 1- \frac{1}{n\log n} + \frac{i}{n^{2}}$,
\item
{\em bivalent} if there are adversarial strategies $\cA,\cA'$ extending this state, such that $\Prob(\cH,\cA)> 1-\frac{1}{n\log n} + \frac{i}{n^2}$ and $\Prob(\cH,\cA') < \frac{1}{n\log n} - \frac{i}{n^{2}}$.
\end{itemize}

\noindent An execution that is $1$-valent or $0$-valent is also called {\em uni-valent}. Note that the types are disjoint and cover the whole space of an algorithm's states.

The classical approach is to show that (i) there exists an ambiguous assignment of input bits to processes 
with
the probabilities of deciding $0$ and deciding $1$ 
being
far from $1$, and
(ii) the adversary can keep the algorithm in the ambiguous state regardless of the algorithm's actions for a certain number of rounds.
Step (i) is guaranteed by the following lemma:

\begin{lemma*}[Lemma~\ref{lem:initial-exe} in Appendix~\ref{sec:proof-lower}]
For any synchronous consensus algorithm there exists an initial state, which, if the adversary can control one process, is null-valent or bivalent.
\end{lemma*} 

We now focus on Step (ii).
In order to obtain amortized analysis of the number of calls to random sources and simultaneously to enforce the sought time lower bound, we give much tighter analysis of this type. In particular, we have to take into account (a) the number of accesses to random sources when analyzing transitions between different types of states, c.f., Lemmas~\ref{lem:null-valent} and~\ref{lem:bi-valent}, and (b) the amortized number of accesses when classifying states in round $i$ according to their valency, c.f., the conditions defining the types of states based~on~valency (see above).

\begin{lemma*}[Lemma~\ref{lem:null-valent} in Appendix~\ref{sec:proof-lower}]
A state $\cH_{i}$ that is null-valent at the beginning of round $i < n$ can be extended to a null-valent state at the end of the round with probability greater than $1 - \frac{2}{n^{2}}$ by failing at most $16\sqrt{r_{i}\log n}$ processes. 
\end{lemma*}

\begin{lemma*}[Lemma~\ref{lem:bi-valent} in Appendix~\ref{sec:proof-lower}]
Let $\cH_{i}$ be a bivalent state. By failing at most $16\sqrt{r_{i}\log{n}}+1$
processes per round, the adversary can extend the state for the next $i' > 1$ rounds, with probability at least $1 - \frac{i'}{n\log{n}}$, reaching a state $\cH_{i + i'}$  that is either bivalent or terminating. The latter case can happen only because failing the necessary processes in round $i + i' - 1$ would exceed the limit $t$ on the total number of failures.
\end{lemma*}

Using the above lemmas, we can prove Theorem~\ref{thm:lower-randomness-res} (equivalent to Theorem~\ref{thm:lower-randomness} in Appendix~\ref{sec:proof-lower}).

\begin{proof}[Proof of Theorem~\ref{thm:lower-randomness-res}]
By Lemma~\ref{lem:initial-exe} in Appendix~\ref{sec:proof-lower} (see also its statement above), the adversary can assign input values such that the initial state is in either a bivalent or a null-valent state. Then, the adversary follows the strategy described in Lemmas~\ref{lem:null-valent} and~\ref{lem:bi-valent} in Appendix~\ref{sec:proof-lower} (see also the statements above), depending whether the current state is null-valent or bivalent.
Specifically, if the state is null-valent, the adversary can extend the execution by one more round with probability at least $1 - \frac{1}{n^{2}}$, again by Lemma~\ref{lem:null-valent} in Appendix~\ref{sec:proof-lower}. If the state is bivalent, it can extend the state for some $i' > 1$ rounds 
with probability at least $1 - \frac{i'}{n\log{n}}$, such that the new state is again either bivalent or null-valent, or the execution terminates but then the number of failed processes in the previous round would exceed the adversary's limit $t$. If the algorithm decides to terminate, it must be in either a $0$-valent or $1$-valent state, since the algorithm is $\left(1 - \frac{1}{n^{3/2}}\right)$-strongly-correct. Therefore, the adversary can prolong the execution either for $n$ rounds or until it runs out of the processes to fail. 

Let $T$ be the round in which the execution terminated. If $T = n$, then the theorem
follows. 
Assume then that the adversary stopped implementing its strategy in round $T$ due the fact that in the preceding round it could not fail the desired number of processes. Since the adversary fails at most 
$16\sqrt{r_{i}\log{n}}+1$ 
processes in a round $i$ (c.f., Lemmas~\ref{lem:null-valent} and~\ref{lem:bi-valent} in Appendix~\ref{sec:proof-lower}), we obtain
\[t \le \sum_{i = 1}^{T - 1}\left(16\sqrt{r_{i}\log{n}}+1\right) 
\le 32 \sum_{i = 1}^{T - 1} \sqrt{(r_{i}+1)\log{n}} \ ,\]
which is equivalent to
\[ t^{2} \le 1024\left(\sum_{i = 1}^{T - 1}\sqrt{r_{i} \log{n}}\right)^{2} \ . \]
Applying the Cauchy-Schwarz inequality to the right hand side of the above, we get
\[ t^2 \le 
1024\left(\sum_{i = 1}^{T - 1}\sqrt{(r_{i} + 1)\log{n}}\right)^{2} 
\le 1024(T - 1)\left(\sum_{i = 1}^{T - 1} (r_{i} + 1)\log{n} \right) \ , 
\]
which, after proper rearranging, yields
\[ \frac{t^{2}}{1024\log{n}} \le (T - 1) \times (R + T) \]
and proves the theorem.
\end{proof}

\section{Interpolation between random and deterministic solutions}
\label{sec:overview-randomness}

To complement the lower bound from Theorem~\ref{thm:lower-randomness-res},
we give a trade-off algorithm that matches it (with respect to poly-logarithmic factor) 
and each access to random source gets only one bit.


\begin{theorem}
\label{thm:trade-off-res}
For any $R \in O(n^{3/2})$, there exists an algorithm solving consensus (with probability~$1$) against the most powerful, adaptive omission-causing adversary failing at most $t < \frac{n}{60}$ processes, which terminates in $\logO\left(\frac{n^{2}}{R}\right)$ rounds and uses, in total, $\logO\left(n^2\right)$ bits of communication and $\logO\left(R\right)$ bits of randomness,~whp.
\end{theorem}

The extended version of the above result is formally described and proved in
Section~\ref{sec:tradeoff-upper}, as Theorem~\ref{thm:alg-tradeoff}. We give here a general overview.
The algorithm uses an idea of grouping processes into smaller subsets, first solving consensus within each subset and then propagating the decision to the entire system. Such a scheme is typically heterogeneous, in the sense that intra-group procedures/consensus mimics an efficient randomized algorithm while the inter-group procedures are based on efficient deterministic solution. Therefore, the smaller the groups the more randomness could be reduced (in theory).
Grouping algorithms were first used in~\cite{DBLP:conf/stoc/HajiaghayiKO22}, where a similar trade-off between \textit{communication/random-bit-complexity} and time was shown in case the processes are prone to crashes only (more benign faults). However, the same framework cannot be used here because, as discussed earlier, omission failures automatically introduce at least quadratic communication complexity, while the scheme in \cite{DBLP:conf/stoc/HajiaghayiKO22} relied on lower-communication
procedures.

Therefore, our omission-based implementation is very different from the crash-based one in~\cite{DBLP:conf/stoc/HajiaghayiKO22} and relies, once more, on the idea of dividing processes into operative and inoperative and on combining this approach with a round-robin algorithm working on groups (i.e., considering groups one-by-one). For the sake of clarity, assume that we split the set of processes $\cP$ into $x$ groups $SP_{1}, \ldots, SP_{x}$, of size at most $\ceil{\frac{n}{x}}$ each. In $x$ subsequent phases, we let each group invoke our almost-optimal consensus algorithm against omissions (limited only to the members of the group), given by Theorem~\ref{thm:omissions-opt-res}, one group per phase. Then, we propagate the decision of the consensus algorithm to other processes in the system and require that each subsequent call to the almost-optimal consensus algorithm, which takes place in any later phase, uses the propagated value as the input bit for the processes involved in the call in this phase. 

The crux is in 
executing
the above only among operative processes. Similarly to the algorithm presented by Theorem~\ref{thm:omissions-opt-res}, we use a result about random graphs of expected degree $\Theta(\log n)$ for specifying the communication graph between groups. We also determine whether a process is operative or not based on the number of messages it received in the random graph. 
Ultimately, we again use the observation that the operative processes, ergo those who constantly receive many messages from processes who are their neighbors in the communication graph, can exchange any information between themselves within $O(\log{n})$ rounds. This property is crucial for correctness of our algorithm and its detailed explanation can be found in the proof of Lemma~\ref{lem:operative-contribution}.
It then follows that if a group, which has sufficiently many non-faulty and operative processes, calculated a decision value, this value is properly distributed to all other operative processes. 
From now on, the operative processes have the same input value in all remaining phases and, therefore, the adversary cannot prevent them from deciding on that value. 
The existence of a group satisfying the requirements on the number of non-faulty and operative processes is a consequence of limiting the adversary to at most a constant fraction of
faults, and -- similarly as in the main Algorithm~\ref{alg:opt-omissions} -- the fact that a random graph of expected degree $\Theta(\log{n})$ has dense and shallow (with a small diameter) sub-graphs of size being a constant fraction of all vertices. 

Now, the trade-off comes from the fact that the number of random bits needed to solve consensus in a group does not scale linearly with the size of the group. If the size of the group is roughly~$\frac{n}{x}$, our optimal algorithm uses $\logO\left(\frac{n}{x}\cdot \sqrt{\frac{n}{x}}\right) = \logO\left(\left( \frac{n}{x}\right)^{3/2}\right)$ random bits per run. It follows 
that performing $x$ of such runs, one by one, gives the desired trade-off between time $T=\logO\left(x\cdot \sqrt{\frac{n}{x}}\right) = \logO\left(\sqrt{nx}\right)$ and the randomness complexity 
$R=\logO\left(x\cdot \left(\frac{n}{x}\right)^{3/2}\right) = \logO\left(\frac{n^2}{\sqrt{nx}}\right) = \logO\left(\frac{n^2}{T}\right)$. 
The communication bit complexity bound follows from applying the bound of Theorem~\ref{thm:main-algo} to each phase of the round-robin scheme separately. 

\section{Future Directions}
\label{sec:future}

We believe that our novel parameterized valency technique could be used in proving lower bounds and impossibility results regarding trade-offs between randomness and other complexity measures in other distributed problems. First important open question is the tight characterization of the tradeoff between round complexity and randomness used by processes for the case when $t = o(n)$. Specifically, Theorem~\ref{thm:alg-tradeoff} in Appendix~\ref{sec:tradeoff-upper} (extended version of Theorem~\ref{thm:trade-off-res}) gives an algorithm satisfying the invariant $\texttt{ROUNDS} \times \texttt{RANDOMNESS} = \tilde\Theta(n^2)$, regardless of the number of faults. On the other hand, Theorem~\ref{thm:lower-randomness-res} gives only $\tilde\Omega(t^2)$ lower bound on this product.

Second important open question is about trade-off between communication complexity and time when $t=o(n)$, for any type of faults. In particular, for crash failures there is only an upper bound known~\cite{DBLP:conf/stoc/HajiaghayiKO22} for such tradeoff, while for omission failures, surprisingly, there is no such tradeoff when $t = \Theta(n)$, by our Theorem~\ref{thm:omissions-opt-res} (as it simultaneously matches, up to a polylogarithmic factor, the two independent lower bounds on time~\cite{Bar-JosephB98} and communication~\cite{AbrahamCDNPRS19}).

The concept of operative processes, maintaining them locally at (relatively) low cost and using them for performing tasks such as efficient counting and information exchange, could be a game-changing concept in designing and analysis of distributed fault-tolerant algorithms. We demonstrated how to use it efficiently against adaptive omission faults, but we suspect it could be applied to other computation and failure models and problems. On the technical side, it would be interesting to further improve communication performance of using operative processes in case of smaller number of failures.

\bibliographystyle{plain}
\bibliography{bibliography}

\appendix

\begin{center}
    {\bf\Large Appendix}
\end{center}

\section{Other Related Work}
\label{sec:related}

The consensus problem was introduced by Lamport, Pease and Shostak~\cite{LamportSP82,PeaseSL80}, in the context of {\bf\em deterministic} solutions. Fischer, Lynch and Paterson~\cite{FischerLP85} showed that the problem is unsolvable by any deterministic algorithm in an asynchronous setting, even 
if one process may fail.
Fischer and Lynch~\cite{FischerL82} proved that a synchronous solution requires $t+1$ rounds if up to~$t$~processes~may~crash, which is automatically applicable to settings with more severe failures, including omissions.
Dolev and Strong~\cite{DolevS83} gave an efficient deterministic solutions to Consensus under even a stronger Authenticated Byzantine failures, working in optimal time $t+1$ and using $\Omega(nt)$ messages, which was (nearly) matched by Dolev and Reischuk~\cite{DolevR85}, 
who proved an $\Omega(t^2+n)$ lower bound.

Recently, Abraham et al.~\cite{AbrahamCDNPRS19} showed that any (even {\bf\em randomized}) algorithm that solves Consensus with a constant probability against an adaptive adversary requires $\Omega(t^2 + n)$ messages, 
whp. 
Regarding time complexity, Bar Joseph and Ben-Or~\cite{Bar-JosephB98} showed that no algorithm can solve Consensus, with correctness probability $1$, against an adaptive adversary in fewer than $\Omega(\frac{t}{\sqrt{n\log{n}}})$ expected rounds. It holds even for more benign crash failures. Other classic work~\cite{karlin1986probabilistic} shows a lower bound on the probability of failure of any algorithm that runs in a fixed number of rounds, given that a linear (in $n$) number of processes is prone to Byzantine failures.
This result holds even if the adversary's knowledge is limited to the history of controlled processes and if the algorithm has cryptographic routines such as the public-key infrastructure or threshold signatures. Following a long line of work~\cite{AbrahamCDNPRS19, fitzi2003efficient, ghinea2022round, katz2006expected, micali2017optimal}, this result has been recently matched, under these cryptographic assumptions, by an algorithm of~\cite{ghinea2022round}. 

On the algorithms' side, the deterministic algorithm by Dolev and Strong~\cite{DolevS83} has been also the best known overall solution (including randomized algorithms) against an adaptive adversary (of unbounded computational power and full information) causing omission failures, with respect to time and message complexity. Other works that consider omission failures, sometimes known as general omission failures (since both, incoming and outgoing messages of a faulty processes can be omitted), are mostly concerned with the early-stopping version of the consensus problem variant~\cite{raipin2010strongly, rocsu1996early}, unlike our work focused on optimization of three complexity measures.

Better time and/or communication complexity is possible if {\bf\em failures are more benign,} e.g., processes are crashed permanently (instead of message omissions) or the {\bf\em adversary is weaker} (more restricted or more oblivious).
Hajiaghayi et al.~\cite{DBLP:conf/stoc/HajiaghayiKO22} designed a solution under $t < n$ crashes in almost-optimal time and using a subquadratic number of communication bits $\logO(n^{3/2})$. 
King and Saia~\cite{KingS11} proved that under some limitations on the Byzantine adversary and requiring termination only whp, the subquadratic expected communication complexity $\logO(n^{3/2})$ and even polylogarithmic time can be achieved. 
Even better complexities could be achieved against an {\em oblivious adversary}, i.e., the adversary who knows the algorithm but has to decide which process to fail and when before the execution starts. 
Chor, Merritt and Shmoys~\cite{ChorMS89} developed constant-time randomized algorithms 
%
while 
Gilbert and Kowalski~\cite{GilbertK10} 
solved Consensus in
optimal communication complexity $O(n)$ 
and 
$ \mathcal{O}(\log n) $ time, whp.
Sub-quadratic communication is also possible in asynchronous setting under oblivious adversary, which decides processes' and messages' delays in advance, c.f.,~\cite{GeorgiouGGK13}.

On the other side, {\bf\em stronger adversaries} (such as general Byzantine, discussed earlier) or/and {\bf\em asynchrony} require linear time or/and (super-)quadratic communication.
Georgiou et al.~\cite{GeorgiouGGK13} showed that even a simpler task of fault-tolerant gossip requires linear time {\em or} quadratic communication when an adaptive adversary crashes processes.
Recently, Alistarh
et al.~\cite{AlistarhAKS18} showed how to obtain almost-optimal communication complexity $O(n^2\log{n})$ if less then $n/2$ processes may fail, which improved upon the previous result $O(n^2\log^2 {n})$ by Aspnes and Waarts~\cite{AspnesW96} and is asymptotically almost optimal due to the lower bound $\Omega(n^2)$ by Attiya and Censor.

\section{Almost-optimal Algorithm against Omissions Failures}
\label{sec:main-algo}

We present a randomized algorithm that solves consensus with probability $1$. It works in $O\left(\frac{t}{\sqrt{n}}\log^{2}{n}\right)$ rounds with high probability and uses $O\left(n(t\log^4{n} + n) \right)$ communication bits with probability $1$ against any number up to $t < \frac{n}{30}$ of omission-faulty processes. We give a description below while the pseudocode of the algorithm has been provided in Algorithm~\ref{alg:opt-omissions} in Section~\ref{sec:overview-alg}, together with an extended overview and illustrative figures.

The algorithm partitions processes into two classes: \textit{operative} processes and \textit{inoperative} processes. The operative processes have the property that any two of them can exchange a bit of information in at most $O(\log{n})$ rounds, using a certain sparse graph of communication. In contrast, the inoperative processes are those that, due to the actions of the adversary, have been excluded from the fast flow of information guaranteed for the operative ones. 
We note that our partition is dynamic. The status of a process can change during the run of the algorithm. 
Our partition is also independent of the partition into faulty and non-faulty processes at the logical level. A non-faulty process may become inoperative if its communication neighborhood diminishes (e.g., many of its randomly selected neighbors could become faulty and omit messages), while faulty processes may stay operative as long as they are good transmitters from the perspective of other processes, despite the omissions of some of their messages. The operative/inoperative status of a process is also local -- other processes may not be aware of it.

To start, all processes are initially designated as operative (see line~\ref{line:operative-init}). Next, each operative process determines its communication neighbors, a set denoted by $V_{p}$ in the algorithm, based on its identifier $p$ and the set of all processes $\cP$ taking part in the protocol. To be precise, in Theorem~\ref{thm:random-graph-properties} we state the existence of a sparse random graph with a fault-tolerant properties strong enough to support the later stage of the algorithm regardless of the omissions pattern. Since the processes know $\cP$, they can all pre-compute the same graph (i.e., they can all choose a lexicographically smallest such graph) that is guaranteed by Theorem~\ref{thm:random-graph-properties}. Observe that, in order to do this selection, the processes do not need to communicate, as the properties of such graph are purely combinatorial 
and only
rely on the knowledge of $p$ and $\cP$.

During the main loop of the algorithm (lines~\ref{line:main-for-begin}-\ref{line:main-for-end}), the operative processes work collectively to establish a consensus decision among them. This work is divided into $O\left(\frac{t}{\sqrt{n}}\log n\right)$ \textit{epochs} and utilizes a variant of the majority-vote protocols introduced for the case of, more benign, permanent crash failures in~\cite{Bar-JosephB98}. 

In our approach, only operative processes participate in the majority-vote protocol. This is achieved as follows: each operative process prepares a variable (interpreted as a \textit{candidate value} for the final decision) $b_p$ initialized to its input bit. In each epoch, the operative processes rely on two different communication subroutines, \textsc{GroupBitsAggregation} and \textsc{GroupBitsSpreading} (see lines~\ref{line:group-bits-aggr} and \ref{line:sum_ones_zeros}), to calculate the number of operative 1's and 0's present in the candidate values. 
The communication algorithms \textsc{GroupBitsAggregation} and \textsc{GroupBitsSpreading} guarantee that, after $O(\log{n})$ rounds and using at most $O\left( n^{3/2} \polylog(n)\right)$ communication bits, every operative process has a correct estimation of the number of operative 1's in the candidate values from the beginning of the epoch (the variable $\texttt{ones}_{p}$ in the pseudocode) and operative 0's in the candidate values from the beginning of the epoch (the variable $\texttt{zeros}_{p}$ in the pseudocode) after their completion. Since the description of these routines is highly technical, we devoted Subsection~\ref{subsec:comm-patterns} for this purpose.
Based on the values of the variables $\texttt{ones}_{p}$ and $\texttt{zeros}_{p}$ process $p$ decides what candidate value assign to variable $b_{p}$ for the next epoch. If the estimation $\texttt{ones}_{p}$ makes at least $\frac{18}{30}$ fraction of the estimation of all operative processes (which is equal to
$\texttt{ones}_{p} + \texttt{zeros}_{p}$), then the process assigns $b_{p} = 1$; if the estimation $\texttt{ones}_{p}$ is less than half of the estimation of the total number of operative processes then the process assigns $b_{p} = 0$. In any other case, it assigns $b_p$ to a uniform random bit. The proportion $\frac{18}{30}$ is exactly $\frac{1}{2}$ incremented by the maximal possible fraction of inoperative processes there can be in the algorithm. It guarantees that no two operative processes can deterministically assign $b_{p} = 0$ and $b_{p} = 1$ at the same time.
See Figure~\ref{fig:random-coin} for illustration.

In the analysis, we will demonstrate that if the number of new processes becoming inoperative in a number of consecutive epochs is smaller than $\sqrt{n}$, the choices made by the operative processes will lead to the same candidate value $b_p$ across all operative processes with a constant probability. Our communication algorithms are designed in such a way that $n - \Theta(t)$ processes are always operative in the system, regardless of the adversary's actions. Therefore, after $O\left(\frac{t}{\sqrt{n}}\log{n}\right)$ such epochs, we can assert that operative processes have reached a consensus decision (i.e., on same value) with high probability. To lift the probability of success to $1$, the operative processes adopt a safety rule in the following form. If the estimation value $\texttt{ones}_{p}$ makes at least $\frac{9}{10}$ fraction of the estimation of all operative processes (or, symmetrically, the estimation of $\texttt{ones}$ makes less than $\frac{1}{10}$ fraction), a process $p$ marks in an auxiliary variable $\texttt{decided}_{p}$ that it is ready to decide by setting it to $true$. The idea is that, according to the previous arguments, it is only with probability less than $O\left(\frac{1}{n} \right)$ that there is an operative process that has not become ready to decide during the $O\left(\frac{t}{\sqrt{n}}\log{n}\right)$ epochs. In this unlikely event, processes of this type execute a slow but deterministic consensus protocol given as Theorem~4 in~\cite{DolevS83}, whose $O(t)$ running time is accounted to the total time complexity with $O\left(\frac{1}{n} \right)$ probability only.

To be more precise, in the last part of the algorithm, lines~\ref{line:good-spread}-\ref{line:fixing-und}, the operative processes who have the variable $\texttt{decided}_{p}$ set to $true$, broadcast their variable $b_{p}$ (at this time it is a consensus value) to all other processes. Then, the operative processes that have the variable $\texttt{decided}_{p}$ set to $true$ and inoperative processes who have received a consensus value $b$ in the previous round, acquire it, decide on the value and stop from participating in the protocol. 
On the other hand, all operative processes that have the variable $\texttt{decided}_{p}$ set to $false$, execute the deterministic consensus protocol from~\cite{DolevS83}, Theorem~4, which reaches consensus with probability $1$ in $O(t)$ rounds. In the analysis, we show that if any operative process has the variable $\texttt{decided}_{p}$ set to $true$, all other operative processes have the same candidate value in their bit $b_{p}$ as the operative process with the variable $\texttt{decided}_p$ set to $true$, thus -- the decision can safely by made on the value of the variable $b$.
The formal analysis of the correctness, running time and bit complexity is given in Subsection~\ref{subsec:analysis-main}.

\subsection{Communication patterns of operative processes within an epoch}
\label{subsec:comm-patterns}

In this part we describe the communication algorithms used by processes to gather information about candidate values' distribution from the beginning of the epoch. We start by explaining the algorithm \textsc{GroupBitsAggregation}, which achieves the goal limited only to subgroups of processes of size $O(\sqrt{n})$. Limiting the size of a group on which we estimate the candidate values is crucial for achieving $O(n^2)$ communication bit complexity.\\
\noindent \textbf{Explanation of \textsc{GroupBitsAggregation}.} 
We partition the set of processes $\cP$ into at most $\ceil{\sqrt{n}}$ groups $W_{1}, \ldots, W_{\ceil{\sqrt{n}}}$, each with a size of at most $\ceil{\sqrt{n}}$ processes, as shown in lines~\ref{line:sqrt-part}-\ref{line:group}. 
See also an example in Figure~\ref{fig:group-partition}.
Since every process knows the set of process names, the partition could be pre-defined, i.e., can be calculated locally by each process using the same local deterministic algorithm. Processes in different groups work independently in algorithm \textsc{GroupBitsAggregation}.

The objective of each group is to compute the number of operative processes within the group that have the variable $b$ set to $0$ and the number of those with $b = 1$, or mark inoperative the processes that cannot gather the necessary information. At the end of the algorithm, each operative process $p$ should have these two values, $\texttt{b\_zeros}_{p}$ and $\texttt{b\_ones}_{p}$, stored in its local memory, and updated its operative status $\texttt{operative}_{p}$. Algorithm~\ref{alg:bits-agreg} presents the pseudocode, while a formal description is provided below.
For each group $W_{i}$, where $1 \le i \le \ceil{\sqrt{n}}$, we use a binary-tree decomposition of depth $\ceil{\log(\sqrt{n}))}$. At the lowest layer $L^{(i)}(1)$ of the decomposition, the processes in $W_{i}$ are divided into $\ceil{\sqrt{n}}$ singleton \textit{bags}, yielding the following:
$$L^{(i)}(1) := \{ L^{(i)}(1,1), \ldots, L^{(i)}(1, \ceil{\sqrt{n}} \}.$$ 
For any subsequent layer $j$, $2 \le j \le \ceil{\log(\sqrt{n})}$, the bags of layer $j$ are formed as a union of bags at level $j-1$ corresponding to the children of the $j$-level bag in the tree. This way, at the higher levels of the decomposition, the number of bags is halved, but the size of a single bag is doubled. Formally, the bag $L^{(i)}(j, k)$, for $1 \le k \le \ceil{\frac{\sqrt{n}}{2^j}}$, which belongs to layer $L^{(i)}(j)$, is the union of bags $L^{(i)}(j - 1, 2k - 1)$ and $L^{(i)}(j - 1, 2k)$. Here, we assume that $L^{(i)}(j-1, k') = \emptyset$ if $k' > \ceil{\frac{\sqrt{n}}{2^{(j-1)}}}$.

The communication scheme is navigated by the structure described above. It proceeds in $\ceil{\log(\sqrt{n})}$ stages that correspond to aggregating knowledge within bags that belong to increasingly higher levels of the decomposition. 

In the first stage, each operative process $p$ that is a member of a bag $L^{(i)}(1, k)$ at the lowest layer $L^{(i)}(1)$ (for $1 \le k \le \ceil{\sqrt{n}}$) initializes two variables: $\texttt{b\_ones}_{p}(1,k)$ and $\texttt{b\_zeros}_{p}(1,k)$ that keep track of the value of its bit $b_{p}$, i.e., $\texttt{b\_ones}_{p}(1,k) = 1$ if $b_{p} = 1$, and $\texttt{b\_zeros}_{p}(1,k) = 1$ if $b_{p} = 0$; inoperative processes skip this step as their values $b$ are not counted  (see lines~\ref{line:if-operative-init}-\ref{line:if-inoperative-init}). Since at level $1$ the bags are singletons, every operative process has the correct number of one and zero candidate values of other operative processes in its bag.

In stage $j$, for $2 \le j \le \ceil{\log(\sqrt{n})}$ processes work to gather information of the candidate values of all other operative processes that belong to the same bag in the layer $L^{(i)}(j)$, see the main loop in lines~\ref{line:main-for-group}-\ref{line:main-for-group-end}. For $1 \le k \le \ceil{\frac{\sqrt{n}}{2^j}}$, fix a bag $L^{(i)}(j, k)$ of the layer $L^{(i)}(j)$. By the hierarchical structure of the tree, at this point of the execution, the operative processes belonging to the left-child bag of $L^{(i)}(j - 1, 2k - 1)$ already collected the number of operative $0$'s and $1$'s among them. The same is true for processes belonging to the right-child bag $L^{(i)}(j - 1, 2k)$. 
To distribute this information between the operative processes in the bag $L^{(i)}(j, k)$, they use procedure $\textsc{GroupRelay}$~described~below.\\

\noindent \textbf{Specification of the procedure \textsc{GroupRelay}.}
The procedure takes five inputs: the identifier of a process $p$ that executes it, its operative status $\texttt{operative}_{p}$, its group $W_{i}$, the bag $L^{(i)}(j, k)$ 
in which $p$ currently exchanges information, and the set $\texttt{counts}_{p}$, which is the information to be distributed among other operative processes in $L^{(i)}(j, k)$. The information could be initialized differently by the algorithm calling the procedure \textsc{GroupRelay}, depending on whether $p$ belongs to a left-child or a right-child of the bag $L^{(i)}(j, k)$ in the lower layer $j$, as shown in lines~\ref{line:group-relay-1}-\ref{line:group-relay-2} of Algorithm~\ref{alg:bits-agreg}.

During the execution of the procedure \textsc{GroupRelay}, the processes in $W_{i}$ have two different roles: \textit{non-faulty} processes of the entire group serve as \textit{transmitters}, while \textit{operative} processes in the bag $L^{(i)}(j, k)$ have an additional role of a \textit{source}. The following $3$-round protocol is executed. In the first round, each source sends its set $\texttt{count}$ to all transmitters. The transmitters collect received sets $\texttt{count}$ and prepare new sets of the same name $\texttt{count}$ that are the union of the former. There are at most four logically different values that processes (sources) of the bag $L^{(i)}(j,k)$ try to disseminate, which are the counts of operative zeros and ones originating in the two bags that are children of the bag $L^{(i)}(j,k)$ in the lower level of the decomposition. In case there are two or more different counts of operative zeros and ones for a children bag, the transmitters may choose arbitrary of them. Therefore, each newly prepared set has the size of at most $O(\log{n})$ bits.
In the second round, the transmitters confirm to sources if they received a message in the previous round. Any source that receives less than $\frac{1}{2}|W_{i}| + 1$ confirmations becomes inoperative. In the third round, the transmitters send the new sets $\texttt{count}$ to the sources. Again, the sources that receive less than $\frac{1}{2}|W_{i}| + 1$ notifications in this round become inoperative. If a process becomes inoperative, it stops serving as a source, although it is still considered a potential transmitter in subsequent calls of the algorithm. The change in the operative status of a process is reported to the main algorithm \textsc{OptimalOmissionsConsenus} via variable $\texttt{operative}_{p}$. 
Example of how the above $3$-round relay process works is given in Figure~\ref{fig:single-group-tree}. 

After the procedure \textsc{GroupRelay} terminates, each process adds the received counts of operative ones and zeros and proceeds to the next stage of \textsc{GroupBitsAggregation}. Therefore, after the last stage, the operative processes gather the knowledge about the number of operative processes having $b=0$ and the number of those who have $b=1$ in the entire group $W_{i}$ (as this set corresponds to the bag being the root of the tree). These numbers are stored in variables $\texttt{b\_zeros}_{p}(\ceil{\sqrt{\log{n}}}, 1)$ and $\texttt{b\_ones}_{p}(\ceil{\sqrt{\log{n}}}, 1)$, and together with the operative status of a process, they are returned to the main algorithm \textsc{OptimalOmissionsConsenus}.

\begin{algorithm}
\SetAlFnt{\tiny}
\SetAlgoLined
\SetKwInput{Input}{input}
\Input{$W_{i}, p,  \texttt{operative}_{p}; b_{p}$}
\If{$\texttt{operative}_{p} = true$\label{line:if-operative-init}}
{
\lIf{$b_{p} = 0$} {$\texttt{b\_ones}_{p}(1, i) \leftarrow 0, \texttt{b\_zeros}_{p}(1, i) \leftarrow 1$}
\lElse{$\texttt{b\_ones}_{p}(1, i) \leftarrow 1, \texttt{b\_zeros}_{p}(1, i) \leftarrow 0$\label{line:if-inoperative-init}}
}

\For{stage $j \in \{2, \ldots, \ceil{\log{n}}\}$\label{line:main-for-group}}
{
    let $L^{(i)}(j, k)$ be $p$'s bag in $j$-th layer $L^{(i)}(j)$\; 
    \lIf{$p \in L^{(i)}(j - 1, 2k - 1)$ \label{line:group-relay-1}}
    {
    $\texttt{counts}_{p} \leftarrow \{\texttt{b\_ones}_{p}(j - 1, 2k - 1), \texttt{b\_zeros}_{p}(j - 1, 2k - 1) \}$
    }
    \lElse{
        $\texttt{counts}_{p} \leftarrow \{\texttt{b\_ones}_{p}(j - 1, 2k), \texttt{b\_zeros}_{p}(j - 1, 2k) \}$ \label{line:group-relay-2}
    }
    $\texttt{counts}_{p}, \texttt{operative}_{p} \leftarrow \textsc{GroupRelay}(W_{i}, p,\texttt{operative}_{p}, L^{(i)}(j,k), counts_{p})$\;
compute values of variables $\texttt{b\_ones}_{p}(j - 1, 2k - 1), \texttt{b\_zeros}_{p}(j - 1, 2k - 1), \texttt{b\_ones}_{p}(j - 1, 2k)$, \ \ \ \ \ \ \ \ \ $\texttt{b\_zeros}_{p}(j - 1, 2k)$ based on the elements in the set $\texttt{counts}_{p}$\;
    $\texttt{b\_ones}_{p}(j, k) \leftarrow \texttt{b\_ones}_{p}(j - 1, 2k - 1) + \texttt{b\_ones}_{p}(j - 1, 2k)$\; 
    $\texttt{b\_zeros}_{p}(j, k) \leftarrow \texttt{b\_zeros}_{p}(j - 1, 2k - 1) + \texttt{b\_zeros}_{p}(j - 1, 2k) $
    \label{line:main-for-group-end}
}
\textbf{return} 
{$\texttt{b\_ones}_{p}(\ceil{\log{n}}, 1), \texttt{b\_zeros}_{p}(\ceil{\log{n}}, 1), \texttt{operative}_{p}$}\label{line:aggr-return}
\caption{\textsc{GroupBitsAggregation}\label{alg:bits-agreg}}
\end{algorithm}

\noindent \textbf{Explanation of \textsc{GroupBitsSpreading.}}
The algorithm \textsc{GroupBitsSpreading} serves as a tool to disseminate the values calculated by the grouped operative processes (in the run of the algorithm \textsc{GroupBitsAggregation}) to all operative processes.
The processes locally prepare an array of pairs $\texttt{BitPacks}$ of size $\ceil{\sqrt{n}}$ as follows. If $W_{\ell}$ is the group of a process in the partitioning $W_{1}, \ldots, W_{\ceil{\sqrt{n}}}$, then the process writes the numbers of operative processes of $W_{\ell}$ having $b = 0$ and having $b = 1$, which is the result of the algorithm \textsc{GroupBitsAggregation}, at the position $\ell$ of the array, keeping all other entries empty. 
Next, the communication is performed along links corresponding to the predetermined graph $G$ with certain combinatorial properties, which existence is guaranteed in Theorem~\ref{thm:random-graph-properties}, cf., line~\ref{line:graph_sampling} of Algorithm~\ref{alg:opt-omissions}. 
Recall that $G$ is such that every process $p \in \cP$ has degree $O(\log{n})$. The set $V_{p}$ of neighbors is passed as a parameter to every execution of the algorithm \textsc{GroupBitsSpreading}.
See also Figure~\ref{fig:group-partition} for an illustration of the group partitioning with additional graph spanned on nodes.

The communication in the \textsc{GroupBitsSpreading} algorithm lasts $8\log{n}$ rounds. Initially, $p$ sends to every process in $V_{p}$ the entries of the array $\texttt{BitPacks}_{p}$ that have not been sent via the link yet. When receiving a message from another process $q \in V_{p}$ process $p$ updates its array $\texttt{BitPacks}_{p}$ by replacing empty entries with the non-empty entries of the array $\texttt{BitPacks}_{q}$. Additionally, $p$ keeps track of the processes from $V_{p}$ who failed to deliver a message in the most recent round and excludes them from the communication and refutes to accept messages from them in any future round of the algorithm \textsc{GroupBitsSpreading}. 
If the number of processes from which $p$ received a message is less than $\Delta / 3$, where $\Delta$ is the parameter from Theorem~\ref{thm:random-graph-properties}, $p$ becomes inoperative. It stays idle to the end of the current epoch and in all future epochs of the execution of the main algorithm \textsc{OptimalOmissionsConsensus}.

After $8\log{n}$ rounds, each operative process adds up the values in its array $\texttt{BitPacks}$ corresponding to the numbers of $0$'s and $1$'s, resp., in different groups, and returns these two numbers. The key property, which we prove formally in the analysis of the protocol, is that any two values returned by the operative processes (either the count of $0$'s or $1$'s) may differ by at most the number of processes that have turned inoperative during the epoch. 

\begin{algorithm}
\SetAlFnt{\tiny}
\SetAlgoLined
\SetKwInput{Input}{input}
\Input{$V_{p}, p, \ell, \texttt{operative}_{p}; \texttt{g\_ones}_{p},\texttt{g\_zeros}_{p}$}
$\texttt{BitPacks}_{p} \leftarrow [(\emptyset, \emptyset), \ldots, (\emptyset, \emptyset)]$, an array of pairs (initially of zeros) of size $\ceil{\sqrt{n}}$\;
$\texttt{BitPacks}_{p}[\ell] \leftarrow (\texttt{g\_ones}_{p}, \texttt{g\_zeros}_{p})$\;
\For{$8\log{n}$ rounds}
{
    \textbf{for} $q \in V_{p}$ : \newline
    \hspace*{2.5mm} \textbf{send} a 
    message to $q$ containing the entries of $\texttt{BitPacks}_{p}$ not shared with $q$ before, provided $q$ has not been disregarded earlier\newline 
    \hspace*{2.5mm} \textbf{receive} a message (if any) with $\texttt{BitPacks}_{q}$ sent by $q$:\newline
    \hspace*{5mm} \textbf{if} $q$ has not sent a message, disregard sending to $q$ in any future round;\newline
    \hspace*{5mm} \textbf{for} $i \in [\ceil{\sqrt{n}}]$ :  
    $\texttt{BitPacks}_{p}[i] \leftarrow \texttt{BitPacks}_{p}[i] \vee \texttt{BitPacks}_{q}[i]$\;\label{line:spreading-comm}
    \lIf{number of received messages is less than $\Delta / 3$}{\newline \hspace*{2.5mm} $\texttt{operative}_{p} \leftarrow false$;\newline \hspace*{2.5mm} stay idle until line~\ref{line:spread-return}}
}
$\texttt{ones}_{p}, \texttt{zeros}_{p} \leftarrow $ the sum of the first (or the second resp.) elements of each non-empty pair of $\texttt{BitPacks}_{p}$\;
\textbf{return} $\texttt{ones}_{p}, \texttt{zeros}_{p}, \texttt{operative}_{p}$\label{line:spread-return}
\caption{\textsc{GroupBitsSpreading}\label{alg:bits-spread}}
\end{algorithm}


\subsection{Analysis of algorithm \textsc{OptimalOmissionsConsensus}}
\label{subsec:analysis-main}

The building blocks of an execution of the algorithm are \textit{epochs} -- full iterations of the main loop of Algorithm~\ref{alg:opt-omissions}, \textsc{OptimalOmissionsConsensus}. 
As any final decision of a process comes from a communication with an operative process, see lines~(\ref{line:final-receiv-1},~\ref{line:final-receiv-2}), our goal is to prove that, with probability 1, the operative processes store the same value in variables $b_p$. 
To achieve that, we will gradually analyze how consecutive subroutines executed within an epoch influence the values $b_p$ held by operative processes.\\
\noindent \textbf{Analysis of} \textsc{GroupBitsAggregation} \textbf{algorithm.}
In this algorithm, the operative processes work in groups that are given by the partition $\{ W_{1}, \ldots, W_{\ceil{\sqrt{n}}}\}$ of the set $\cP$.
Consider a group $W_{i}$.  For a process $p$ and a process $q$ such that $p, q$ in $W_{i}$, we say that $q$ \textit{contributes} to variable $\texttt{b\_ones}_{p}(j, k)$ $(\texttt{b\_zeros}_{p}(j, k)$ respectively), for some $1 \le j \le \ceil{\log(\sqrt{n})}$, $k \ge 1$, if the input $b_{q}$ is counted in the variable $\texttt{b\_ones}_{p}(j, k)$ (or $(\texttt{b\_zeros}_{p}(j, k)$ depending on the value of $b_{q}$). 

\begin{lemma}\label{lem:bits-agg-contr}
Let $M_{i} \subseteq W_{i}$ be the set of operative processes in the group $W_{i}$ at the end of the algorithm \textsc{GroupBitsAggregation}.
Any process $p \in M_{i}$ contributes to both variables, $\texttt{b\_ones}_{q}(\ceil{\log(\sqrt{n})}, 1)$ and $\texttt{b\_zeros}_{q}(\ceil{\log(\sqrt{n})}, 1)$, of any other process $q \in M_{i}$.
\end{lemma}

\begin{proof}

In the proof we inductively show that any two processes $p', q'$ belonging to the intersection of a bag $L^{(i)}(j, k)$ and the set $M_{i}$ have the property that $p'$ contributes to variables $\texttt{b\_ones}_{q'}(j, k)$, $\texttt{b\_zeros}_{q'}(j, k)$ of the process $q'$, where $1 \le j \le \ceil{\log(\sqrt{n})}$, $k \ge 1$. 
The induction proceeds with respect to the depth of a layer, $j$, in the tree decomposition $L^{(i)}$ of the group $W_{i}$. 
For the base case, we observe that the lemma holds for any pair of processes belonging to the set $M_{i} \cap L^{(i)}(1, k)$, where $L^{(i)}(1, k)$ is a bag in the first layer, simply because $M_{i} \cap L^{(i)}(1, k)$ is a singleton set (in this case, the only feasible pair has $p' = q'$).

For the inductive step, consider a bag $L^{(i)}(j, k)$ in some layer $j > 1$. Let $p'$ and $q'$ be any two processes belonging to $M_{i} \cap L^{(i)}(j, k)$. During the execution of the $\textsc{GroupRelay}$ procedure, the process $p'$, who had the role of a source, relayed its set $\texttt{count}$ to at least $\frac{1}{2}|W_{i}| + 1$ other processes (transmitters), because otherwise it would have lost its operative status. Similarly, $q'$ has received the accumulated set $\texttt{count}$ from at least another $\frac{1}{2}|W_{i}| + 1$ processes. However, these two sets of transmitters must have at least one common transmitter, who in this case assures that $q'$ received the counts of $p'$. 
Since $L^{(i)}(\ceil{\log(\sqrt{n})}, 1) = W_{i}$, the lemma is proven.
\end{proof}

Next, we bound the bit complexity of the algorithm $\textsc{GroupBitsSpreading}$ when limited to a single group only.

\begin{lemma}\label{lem:msg-aggr}
Processes in a single group use at most $O(n\log^2{n})$ bits of communication, in total, in an execution of the algorithm $\textsc{GroupBitsAggregation}$.
\end{lemma}

\begin{proof}
The algorithm takes $O(\log{n})$ stages of communication. In a stage, the procedure $\textsc{GroupRelay}$ is executed to distribute information between processes belonging to every bag of the layer of the stage. The information about a single bag distributed in the messages of the procedure $\textsc{GroupRelay}$ is of size $O(\log{n})$ bits at most, as it contains at most four counts of size $O(\log{n})$ each, c.f. the description of the procedure. The messages send in the procedure are always between transmitters and sources. This first class of processes contains at most $O(\sqrt{n})$ processes (the size of any group). The second amortizes to $O(\sqrt{n})$ if considered the union of all bags on the layer corresponding to the stage. It follows then, that at most $O(n)$ messages is sent per stage, in total. Thus the lemma follows.
\end{proof}

\noindent \textbf{Analysis of} \textsc{GroupBitsSpreading} \textbf{algorithm.}
We first explain formally the scheme of communication used in the algorithm \textsc{GroupBitsSpreading}. Let $\cR(n, \rho)$ be a random graph on $n$ vertices in which every edge is contained independently with probability $\rho$. For suitable parameters $\alpha$~and~$\ell$, the following properties of certain graphs $\cR(n, \rho)$ will be used to analyze omission-tolerance of the predetermined graph $G$ used for communication. 

\begin{definition}[\textit{of expansion and edge-sparsity}]
\hfill
\begin{itemize}
    \item \textbf{Expansion:} A graph~$G$ is said to be  \emph{$\ell$-expanding}, or to be an \emph{$\ell$-expander}, if any two subsets of~$\ell$ nodes are connected by an edge. 
   
    
    \item \textbf{Edge-sparsity:} A graph~$G$ is said to be \emph{$(\ell,\alpha)$-edge-sparse} if for any set $X\subseteq V$ of \emph{at most} $\ell$ nodes, there are at most $\alpha |X|$ edges internal for~$X$.
        
\end{itemize}
\end{definition}

To assure that there exists a graph $G$ that can be used by processes to perform the communication in the algorithms we provide the following theorem. 

\begin{theorem}
\label{thm:random-graph-properties}
Let $n \in \textbf{N}$, 
and $\Delta := 832\log{n}$. A random graph $\cR(n, \Delta / (n - 1))$ satisfies all the below properties whp:\\
\noindent  
	\emph{(i)} it is $(n/10)$-expanding, 
\hspace*{6.4em}\\
\noindent \emph{(ii)} it is $(n/10, \Delta / 15)$-edge-sparse,

\noindent \emph{(iii)} the degree of each node is between $\frac{19}{20}\Delta$ and $\frac{21}{20}\Delta$.	
\end{theorem}

\begin{proof}
In this proof we will use Chernoff bounds in the following form, where $\varepsilon>0$ and $X=\sum_{1 \leq i \leq k} X_{i}$ is a sum of independent Bernoulli trials with the expected value $\mu$:

\begin{equation}
\operatorname{Pr}[X \geq(1+\varepsilon) \mu]<\left(\frac{e^{\varepsilon}}
{(1+\varepsilon)^{1+\varepsilon}}\right)^{\mu} 
\end{equation}
\begin{equation}
\operatorname{Pr}[X \geq(1+\varepsilon) \mu]<e^{-\mu \varepsilon^{2} / 3}, \text { for } \varepsilon \leq 1 
\end{equation}
\begin{equation}
\operatorname{Pr}[X \leq(1-\varepsilon) \mu]<e^{-\mu \varepsilon^{2} / 2}, \text { for } \varepsilon<1
\end{equation}
see~\cite{mitzenmacher2017probability}.

We first prove that property (iii) is satisfied by graph $H$ whp. The expected node degree in $H$ is $\Delta$. By the Chernoff bounds (2) and (3), each degree of $H$ is between $(19 / 20) \Delta$ and $(21 / 20) \Delta$ with probability at least
$$
1-2 \cdot \exp \left\{-\Delta \cdot(1 / 20)^{2} / 3\right\} \geq 1-1 / n
$$
as $\Delta>832 \log n$.

Now we consider property (ii). Consider a set $X_{1} \subseteq V$ of at most $n / 10$ nodes. We show that it has at most $\left|X_{1}\right| \Delta / 15$ internal edges with probability at least $1-1 / n^{2}$. Let $x_{1}$ stand for $\left|X_{1}\right|$. Recall that the expected number of edges internal for $X_{1}$ is $\Delta / (n - 1) \cdot  x_{1}\left(x_{1}-1\right) / 2$. The probability that the number $N$ of internal edges is more than $x_{1} \Delta / 15$ can be upper bounded by the Chernoff bound (2) with $\varepsilon_{1}=1/3$ to obtain
$$
\operatorname{Pr}\left[N > \left(1 + \varepsilon_{1}\right) \cdot \frac{\Delta x_{1}\left(x_{1}-1\right)}{2 (n - 1)}\right] \leq e^{-\frac{\Delta x_{1}\left(x_{1}-1\right)}{2(n - 1)} \cdot\left(\varepsilon_{1}\right)^{2} / 3} \leq e^{-\Delta x_{1} / 270},
$$
as $x_1 \le n / 10$. It follows that the probability that a set $X_1$ of at most $x_{1}$ nodes with more than $x_{1} \Delta / 15$ internal edges exists in graph $H$, where $x_{1} \le n/10$, is at most
$$
\sum_{x_{1}=1}^{n/10}\left(\begin{array}{c}
n \\
x_{1}
\end{array}\right) e^{-\Delta x_{1} / 270} \leq \sum_{x_{1}=1}^{n / 10} 2^{x_{1} \log(n)- \Delta x_{1}/ 189} \leq \sum_{x_{1}=1}^{n / 10} 2^{-3 x_{1} \log n} \leq 1 / n^{2},
$$
as $\Delta > 832\log{n}$.

It remains to prove property (i). 
Consider any two disjoint sets $X$ and $Y$ of exactly $\ell$ elements each. The expected number of edges connecting two disjoint sets of nodes $X$ and $Y$ is $\Delta |X||Y| / (n - 1)$. Therefore, by the Chernoff bound (3), we obtain that the number $N$ of edges connecting the sets $X$ and $Y$ is less than $2/3 \Delta |X||Y| / (n - 1)$ with a probability that is at most
$$
\operatorname{Pr}[N<(1-1 / 3) \cdot \Delta |X||Y| / (n - 1)] \leq e^{-\Delta|X||Y| / (n - 1)  \cdot (1 / 3)^{2} / 2} \leq e^{-\Delta n / 1800}
$$
Consider an event that two sets of nodes of $n / 10$ elements each and with less than $2/3 \Delta |X||Y| / (n - 1)$ edges between them exist. The probability of this event is at most
$$
\left(\begin{array}{l}
n \\
\frac{n}{10}
\end{array}\right)\left(\begin{array}{l}
n \\
\frac{n}{10}
\end{array}\right) e^{-\Delta n / 1800} \leq 2^{2 n/10 \cdot \lg n} \cdot e^{-\Delta n / 1800} \leq 2^{n / 5 \log n} \cdot 2^{-832\Delta n / 1800} \leq 2^{-2 n\log(n) / 15}<1 / n^{2}.
$$
Since the complement of this event guarantees that every two disjoint set of size $n / 10$ have at least $2/3 \Delta |X||Y| / (n - 1) > 1$ connecting edges, thus $(n/10)$-expansion follows.
\end{proof}

\remove{
The argument proving this theorem follows the probabilistic method, and due to its combinatorial technicality and not direct connection with the main analysis of the algorithmic part, it is deferred to Appendix~\ref{sec:appendix}. The same applies to the proofs of the next Lemmas~\ref{lemma:sparse-to-dense} and~\ref{lem:compact-comp}, which infer further combinatorial properties of graph $G$.\footnote{%
\jo{We note here that statements of similar results to Theorem~\ref{thm:random-graph-properties} and Lemma~\ref{lemma:sparse-to-dense} can be found in~\cite{ChlebusKS-PODC09}, however, the theorem is given without a proof, and the lemma is proven for different constants. To make to paper self-contained, we give the proofs in Appendix~\ref{sec:appendix}.}
}
}

Let us now fix a particular graph $G$ that satisfies the above properties and let this graph be the same as the one on which the processes decide when executing line~\ref{line:graph_sampling}. As a consequence of these properties, it can be also shown that close neighborhoods of any vertex in $G$ contain regular subgraphs and grow rapidly. Let us start by giving the appropriate definition.

\begin{definition}[\textit{of $(\gamma,\delta)$-dense-neighborhood}]
For a node $v\in V$, denote $N^{d}_G(v)$ the set of vertices at distance at most $d$ from $v$ in graph $G$. Then, for $\gamma \ge 1$, a set $S\subseteq N^{\gamma}_G(v), v \in S$ is said to be \emph{$(\gamma,\delta)$-dense-neighborhood for~$v$} if each node in~$S\cap N^{\gamma-1}_G(v)$ has at least $\delta$ neighbors in~$S$.
\end{definition}

To formalize the meaning of the rapid growth of a close neighborhood of a vertex in $G$, we provide the following lemma. 

\begin{lemma} \label{lemma:sparse-to-dense}		
If graph~$G=(V,E)$ of $n$ nodes is $(n/5, \Delta / 15)$-edge-sparse then any $(\gamma,\Delta/3)$-dense-neighborhood for a node $v\in V$ has at least $\min\{2^{\gamma}, n/10\}$ nodes, for $\gamma \ge 0$.
\end{lemma}

\begin{proof}
Consider a node $v \in V$ and a set $S$ which is a $(\gamma, \Delta/3)$-dense-neighborhood for $v$. Let
$A_i$ stand for $S \cap N^i_G(v)$. We show the following fact by induction on $i$, where $1 \le i \le \gamma - 1$:
the set $A_i$ has at least $\min\{2^i, n/10\}$ elements. 
The base of induction for $i = 1$ follows
directly from the property that $v \in S$ has at least $\frac{\Delta}{3} > 3$ neighbors in $A_1 = S \cap N^1_G(v)$. 

The inductive step: suppose that the inequality $|A_i| = |S \cap N^i_G(v)| \ge \min\{2^i, n/10\}$ holds for $i$,
where $1 \le i < \gamma - 1$. We prove the inequality for $i + 1$ as follows. If $A_i$ has at least $n/5$
elements, then we are done. Otherwise the set $A_i$ has at most $|A_i|\Delta/15$ internal edges, by the
$(n/10, \Delta/15)$-edge-sparsity. We obtain that there are at least $|A_i|\Delta/3 - 2\cdot |A_i|\Delta/15 = 3|A_i|\Delta/15$
edges between the sets $A_i$ and $N^{i+1}_{G}(v) \setminus A_i$. Let $X$ denote the set of neighbors of $A_i$ in $N^{i+1}_{G}(v) \ A_i$. It follows that the number of edges internal for $A_{i+1} = A_i \cup X$ is at least $3|A_i|\Delta/15$. On the other hand, this number is at most $(|A_i|+|X|)\Delta/15$, by $(n/10, \Delta/15)$-edge-sparsity. Consequently, the inequality $(|A_i| + |X|)\Delta/15 \ge 3|A_i|\Delta/15$ holds. Hence $|X| \ge 2|A_i|$ and
there are at least $2^{i+1} \ge \min\{2^{i+1}, n/10\}$ elements in $A_{i+1}$ by the inductive assumption.
We use the fact that $A_{\gamma-1} = S \cap N^{\gamma-1}_G (v)$ has at least $\min\{2^{\gamma-1}, n/10\}$ elements. As the inequality $2^{\gamma-1} \ge n/5$ can be verified directly, the proof is complete.
\end{proof}

We note here that statements of similar results to Theorem~\ref{thm:random-graph-properties} and Lemma~\ref{lemma:sparse-to-dense} can be found in~\cite{ChlebusKS-PODC09}, however, the theorem is given without a proof, and the lemma is proven for different constants. Therefore, for the sake of completeness, we provided the proofs of the versions needed in this work.

An important property\footnote{A similar property has been previously observed for graphs of bounded spectral radius, c.f.~\cite{upfal1992tolerating}. There are known results that bounds spectral radius of a random Erdos-Renyi graph, c.f.~\cite{benaych2020spectral, erdHos2013spectral}, and in consequence could lead to the same conclusion. However, for the sake of the completeness, we provide an alternative, fundamental proof that relies solely on combinatorial properties guaranteed by Theorem~\ref{thm:random-graph-properties}.} is that removal of a linear fraction of nodes from the graph $G$ cannot leave many nodes with small neighborhood, as is stated next: 

\begin{lemma}\label{lem:compact-comp}
If graph $G = (V, E)$ satisfies the properties of Theorem~\ref{thm:random-graph-properties}, then for any set of nodes $T \subset G$ of size at most $\frac{n}{15}$ there exists a subset $A$ of $G$ of size at least $n - \frac{4}{3}|T|$ such that:\\
\noindent  
	\emph{(i)} $A \cap T = \emptyset$, 
\hspace*{3.0em}
	\emph{(ii)} every node from $A$ has at least $\frac{\Delta}{3}$ neighbors in $A$.
\end{lemma}

\begin{proof}
Fix any set of nodes $T_{1} \subset V$ of size at most $\frac{n}{10}$ and consider the following, inductive, definition. Given a set $T_{i}$, for $i \ge 1$, we define $T_{i + 1} = T_{i} + \{v\}$ if there exists a node $v \in V \setminus T_{i}$ that has at least $\frac{37}{60}\Delta$ neighbors in $T_{i}$; otherwise define $T_{i+1} = T_{i}$. By the fact that there is a finite number of nodes in $G$ and by the fact that if $T_{j+1} = T_{j}$ then all subsequent sets $T_{j+2}, T_{j + 3}, \ldots $ are the same, there must be exactly one index $K$ such that for all $k \ge K$ it holds $T_{k + 1} = T_{k}$, but $T_{K} \neq T_{K-1}$. 

We prove that $K < |T_{1}| / 3$. 
Assume to the contrary that $K \ge |T_{1}| / 3$. Let $k = |T_{1}| / 3$. 
Based on this assumption, observe that adding a new node $v$ to $T_{i}$, for $1 \le i \le K - 1$, increases the number of edges induced by the set $T_{i}$ by $37\Delta / 60$ at least.
It follows that the number of edges induced by $T_{K}$ is at least $37 K\Delta/60$. 
On the other hand, the number of vertices of $T_{K}$ is $|T_{1}| + k - 1 \le 4/3|T_{1}| = 4/3 \cdot n/15 = 4n/45$. 
Therefore, we can use the $(n/10, \Delta/15)$-edge-sparsity property of $G$ to the set of vertices $T_{k}$ and conclude that there can be at most $4/3|T_{1}| \cdot \Delta / 15 = 4\Delta K / 45$ edges in the subgraph induced by $T_{K}$.
Comparing the last upper bound of $4\Delta K / 45$ on the number of edges with the derived lower bound bound of $37\Delta K / 60$ yields a contradiction as $37\Delta K / 60$ is not less than $4\Delta K /45$. It follows that $K$ must be smaller than $|T_{1}| / 3$.

To finish the proof of the lemma, we define $A$ as $V \setminus T_{K}$. We first observe that degrees of all nodes in $A$ are at least $\Delta / 3$. If there was a node of a smaller degree, it would have to have at least $\frac{19}{20}\Delta - \frac{1}{3}\Delta \ge \frac{37}{60}\Delta$ neighbors in $T_{K}$ and thus it would be added to $T_{K + 1}$ violating the definition of $K$. Since $K < |T_{1}| / 3$, thus $|T_{K}| \le 4|T_{1}|/2$ as there is only one node added at the time. The size of $A$ must be then at least $n - 4|T_{1}|/3$ which completes the proof.
\end{proof}

\noindent Establishing the necessary properties of the communication graph used in runs of the \textsc{GroupBitsSpreading} algorithm, we focus on analyzing how information is spread across the graph in these runs. 
Let us fix a run of the algorithm. For this run, we define $\texttt{OP}_{i}$ to be the set of these processes who maintained their status as operative until the end of the round $i$ of the run. Observe that there is $8\log{n}$ rounds, corresponding precisely to the iterations of the main loop of the algorithm, thus $\texttt{OP}_{8\log{n}}$ denotes the set of processes that finished the run as operative.
We define $G^{\texttt{OP}}_{i}$ a subgraph of $G$ with the set of vertices equal to $\texttt{OP}_{i}$ and the set of edges equal to the links by which processes in $\texttt{OP}_{i}$ received messages in round $i$ of the run of this algorithm. We call a link \textit{operative} in round $i$ if it belongs to $G^{\texttt{OP}}_{i}$. Observe the a link operative in round $i$ correctly transmits messages in both directions in rounds $1, \ldots, i - 1$. That is because links that are observed faulty by a process are never used in the future rounds.


\begin{lemma}\label{lem:operative-dense-neigh}
For a fixed run of the algorithm~\texttt{GroupBitsSpreading}, any processes $p$ in $\texttt{OP}_{8\log{n}}$ has a $(2\log{n}, \Delta/3)$-dense-neighborhood in the graph $G_{8\log{n}}^{\texttt{OP}}$.
\end{lemma}
\begin{proof}
For proving the existence of a $(2\log{n}, \Delta / 3)$-dense-neighborhoods for process $p$, we use the following inductive construction. 
Let $R_{0}(p)$ be the set of these processes from which $p$ receives a message in the last round of the algorithm algorithm \textsc{GroupBitsSpreading}. Observe that $|R_{0}(p)|$ is at least $\Delta/3$ as $p$ is in $\texttt{OP}_{8\log{n}}$. 
Let $B_{1}$ be a subgraph of $G$ with the set of vertices equal to $B_{0} \cup R_{0}(p)$ and the set of edges equal to the links by which $p$ receives a message from $R_{0}(p)$ in the last round of the algorithm. First, we note that the graph $B_{1}$ is a subgraph of $G^{\texttt{OP}}_{8\log{n} - 1}$. Any processes of $B_{1}$ must be operative until round $8\log{n} - 1$ as only operative processes send messages. Also, operative processes use links that continuously deliver messages, c.f.~\ref{line:spreading-comm}, instruction $3$, hence any link that delivers a message \textit{to} process $p$ in round $8\log{n}$ has to deliver a message in rounds $1, \ldots, 8\log{n} - 1$ \textit{from} $p$.
It also holds that, as observed, $|B_{1}| \ge \Delta /3$, and thus the graph $B_{1}$ satisfies the definition of $(1, \Delta / 3)$-dense-neighborhood of $p$ in the graph $G^{\texttt{OP}}_{8\log{n} - 1}$. This concludes the base case of the inductive construction.

For the inductive step, we define the set $R_{i}(p)$, for $2 \le i \le 2\log{n}$, as a set of these processes that send a message to any process in $B_{i-1}$ in the round $8\log{n} - i$ of the algorithm. Let $B_{i}$ be a subgraph of $G$ with the set of vertices equal $R_{i}(p)$ and the set of edges equal to the set of link by which a process from $B_{i}$ received a message in the round $8\log{n} - i$. The same argument as the one used in the base case shows that $B_{i}$ is a subgraph of $G^{\texttt{OP}}_{8\log{n} - i}$. By the inductive assumption $B_{i-1} \subseteq G^{\texttt{OP}}_{8\log{n} - i + 1}$. Also, from the line~\ref{line:spreading-comm} it follows that the operative subgraphs $G^{\texttt{OP}}$ are downward monotonic and thus $G^{\texttt{OP}}_{8\log{n} - i + 1} \subseteq G^{\texttt{OP}}_{8\log{n} - i}$. Therefore $B_{i-1} \subseteq B_{i}$. Also, since $B_{i-1} \subseteq G^{\texttt{OP}}_{8\log{n} - i + 1}$, thus processes from $B_{i}$ are operative until the round $8\log{n} - i + 1 \ge 6\log{n}$. It follows that each of these processes receives at least $\Delta / 3$ messages in the round $6\log{n}$ and these messages must be send from processes in $B_{i}$. Therefore $B_{i}$ satisfies the definition of $(i, \Delta / 3)$-dense-neighborhood of $p$ in $G^{\texttt{OP}}_{8\log{n} - i}$. This completes the induction step and concludes the proof.
\end{proof}
We say that a process $p \in \texttt{OP}_{8\log{n}}$ can \textit{reach} a process $q \in \texttt{OP}_{8\log{n}}$ during the run of the algorithm \textsc{GroupBitsSpreading}, if there exists a path of processes $p = s_{1}, s_{2}, \ldots, s_{k} = q$, such that the process $s_{i + 1}$ received a message from $s_{i}$ in $i$th round of the run of the \textsc{GroupBitsSpreading} algorithm, for any $1 \le i \le k - 1$.

\begin{lemma}\label{lem:spreading-reaching}
For a run of the algorithm~\texttt{GroupBitsSpreading}, any processes $p$ in $\texttt{OP}_{8\log{n}}$ can reach any other process $q$ in $\texttt{OP}_{8\log{n}}$.
\end{lemma}
\begin{proof}
Since $p$ and $q$ belong to $\texttt{OP}_{8\log{n}}$, thus by Lemma~\ref{lem:operative-dense-neigh} there exist $(2\log{n}, \Delta / 3)$-dense-neighborhoods of, respectively, $p$ and $q$ in the graph $G^{\texttt{OP}}_{6\log{n}}$. Denote the neighborhoods $S_{p}$ and $S_{q}$ respectively. Since $G^{\texttt{OP}_{6\log{n}}}$ is a subgraph $G$, thus $S_{1}$ and $S_{2}$ are also subgraphs of $G$. This observation allows us to use Lemma~\ref{lemma:sparse-to-dense} and conclude that the sizes of $S_{1}$ and $S_{2}$ are $n / 10$ at least.
Since there is at most $n / 30$ faulty processes, the are subsets $S^{C}_{1}$, $S^{C}_{2}$ of $S_{1}$ and $S_{2}$ consisting of only non-faulty processes of size at least $n/ 10$ each. By the $(n/10)$-expanding property of $G$ there is at least one edge connecting $S^{C}_{1}$ and $S^{C}_{2}$. Let $p'$ and $q'$ be the endpoint processes of this edge. Now, we can define the path by which $p$ can reach $q$. Processes $p$ and $p'$ belong to $G^{\texttt{OP}}_{6\log{n}}$. By the previous observation links of this subgraph transmit messages bidirectionally in rounds $1, \ldots, 2\log{n}$. Since $p'$ is in a $(2\log{n}, \Delta/3)$-dense-neighborhood of $p$ in $G^{\texttt{OP}}_{6\log{n}}$, thus $p'$ is in distance at most $2\log{n}$ from $p$ in $G^{\texttt{OP}}_{6\log{n}}$. It follows that $p$ can reach $p'$ in $2\log{n}$ rounds. In the round $2\log{n} + 1$ the process $p'$ can reach $q'$ as they are both non-faulty and operative and there is an edge between them. Then, in the next $2\log{n}$ rounds, $q'$ can reach $q$ by the same reasoning as for $p$ and $p'$. Thus, the lemma is proven. 
\end{proof}

\noindent \textbf{Analysis of the main algorithm}.
In the final part of the proof, we connect together the properties of the algorithms \textsc{GroupBitsAggregation} and \textsc{GroupBitsSpreading} and explain how they lead to correct updates of values $b$ in lines~\ref{line:if-1}-~\ref{line:if-random} of the main algorithm~\textsc{OptimalOmissionsConsensus}, resulting eventually in a valid consensus decision.   
We start by noting that the communication graphs used in any of the inner algorithms are dense enough to maintain a large fraction of correct processes operative, regardless of the actions of the adversary. We recall here the assumption that the number $t$ of faulty processes is less than $\frac{n}{30}$.
Also, recall that we call a single iteration of the main loop of the algorithm ~\textsc{OptimalOmissionsConsensus} an epoch. For an epoch, let $\texttt{OP\_END}$ denote the set of these processes that are operative at the end of this epoch.

\begin{lemma}\label{lem:good-proc-are-large}
For any epoch, the size of the set $\texttt{OP\_END}$ is larger than $n - 3t$.
\end{lemma}
\begin{proof}
Consider an epoch $\cE$ of the algorithm and let $F$ be the set of processes that the adversary corrupted in the epoch $\cE$ or before. Recall the partition $\{ W_{1}, \ldots, W_{\ceil{\sqrt{n}}}\}$ of the set $\cP$  used in the algorithm $\textsc{GroupBitsAggregation}$. Since $|F| \le t$, processes then by a counting argument there exists a set of $\sqrt{n} - 2\frac{t}{\sqrt{n}}$ groups, which corresponds to a set of at least $n - 2t$ processes, such that every process in the set is non-faulty and has more than half non-faulty processes in its group. Denote this set $X$.

Consider the subgraph induced by $X$ in the graph $G$ used for the communication in $\textsc{GroupBitsSpreading}$ algorithm. Using Lemma~\ref{lem:compact-comp} for the set $G \setminus X$, which has size at most $2t \le n/15$, we conclude that there exists $X' \subseteq X$ that has size at least $n - 4/3 \cdot 2t = n - 8/3t \ge n - 3t$ such that every process from $X'$ has degree at least $\Delta / 3$ in $X'$.

We finish the proof by arguing, that all process from $X'$ belongs to the set $\texttt{OP\_END}$ for the epoch $\cE$. First, every process $p$ in $X'$ also belongs $X$ thus, by the definition of $X$, it always receives more than half of messages from its group in all the executions of \textsc{GroupBitsAggregation} that occur before or in the epoch $\cE$. 
Second, we choose $X'$ to consist only of non-faulty processes. Communication on the links \textit{between} them is always reliable. By the fact that every processes of $X'$ has degree at least $\Delta / 3$, we conclude that every processes from $X'$ maintains its operative status in all the executions \textsc{GroupBitsSpreading} that occur before or in the epoch $\cE$. Therefore, $X' \subseteq \texttt{OP\_END}$.
\end{proof}
We say that a process $p$ \textit{contributes} to the sum of ones (or zeros) of a process $q$ if its bit $b_{p}$ is included in the value $ones_{q}$ calculated in line~\ref{line:sum_ones_zeros} of the main algorithm~\textsc{OptimalOmissionsConsensus} ($zeros_{q}$ resp.). 
\begin{lemma}\label{lem:operative-contribution}
For any epoch, every process $p$ contained in the set $\texttt{OP\_END}$ with $b_{p} = 1$ contributes to the sum of ones of any other process $q$ from $\texttt{OP\_END}$. Analogically, if $b_{p} = 0$ then $p$ contributes to the sum of zeros of any other process $q$ from $\texttt{OP\_END}$.
\end{lemma}
\begin{proof}
Consider any $p \in \texttt{OP\_END}$. Let $W_{i}$ be the $p$'s group, i.e. $p \in W_{i}$. 
By Lemma~\ref{lem:bits-agg-contr}, we get that any operative process $r$ contributes to the values returned by any other operative process within its group in the algorithm \textsc{GroupBitsAggregation}.
This yields, that whatever pair of values  $(\texttt{g\_ones}_{i}, \texttt{g\_zeros}_{i})$, describing the number of $0$'s and $1$'s among operative process in $W_{i}$, the process $q$ receives during the execution \textsc{GroupBitsSpreading} algorithm, the bit value $b_{p}$ contributes to this pair. 
On the other hand, Lemma~\ref{lem:spreading-reaching} assures that $p$ can reach $q$. 
It follows that $q$ receives at least one pair of values $(\texttt{g\_ones}_{i}, \texttt{g\_zeros}_{i})$, and the lemma is proven.
\end{proof}

\begin{lemma}[Lemma $4.3$ in~\cite{Bar-JosephB98}]\label{lem:anti-concetration}
Assume that $n$ processes independently choose a random bit from uniform distribution. Let $X$ be the random variable denoting the number processes that chose bit $1$. Then for any $t \le \sqrt{n} / 8$
$$\Pr(X - \E(X) \ge t\sqrt{n}) \ge \frac{e^{-4(t+1)^2}}{\sqrt{2\pi}} \ .$$
\end{lemma}

\noindent We call an epoch \textit{good} if at most $\sqrt{n}$ processes become inoperative during the epoch. The following explains how lines~\ref{line:if-1}-~\ref{line:if-random} change $b$ values of operative processes under the assumption of a sequence of consecutive good epochs. 
\begin{lemma}\label{lem:three-epochs}
Consider three consecutive good epochs $\cE_{1}, \cE_{2}, \cE_{3}$. Let $\texttt{OP\_END}_{i}$, for $i \in \{1, 2, 3\}$ be the set of operative processes at the end of the epoch $\cE_{i}$. With probability $\Omega(1)$ all processes belonging to $\texttt{OP\_END}_{3}$ store the same value $b$ before the third epoch ends.
\end{lemma}
\begin{proof}
Using Lemma~\ref{lem:good-proc-are-large}, we have that $|\texttt{OP\_END}_{1}| \ge n - 3t \ge \frac{9}{10} n$, since we assumed  $t < \frac{1}{30}n$. Combining Lemma~\ref{lem:operative-contribution} with the assumption that in the epoch $\cE_{1}$ at most $\sqrt{n}$ change their status to inoperative, we see that values $ones_{p}$, $zeros_{p}$ and $ones_{p} + zeros_{p}$ calculated in a process $p$ in epoch $\cE_{1}$ differ by at most $\sqrt{n}$ from the corresponding values of any other processes from $\texttt{OP\_END}_{1}$. 
Therefore, in the epoch $\cE_{1}$ no two processes can execute line~\ref{line:if-1} and line~\ref{line:if-0} at the same time, since the difference between right-hand-sides of these two inequalities is at least $\frac{1}{30}|\texttt{OP\_END}_{1}| > \sqrt{n}$.

We first consider the case when no process executes line~\ref{line:if-0}, that is processes either assign a random bit to $b$ or $0$. Observe that if a processes assigns $0$ with probability $1$ instead of $\frac{1}{2}$ it is more likely that the total number of zeros exceeds a certain threshold. Therefore by applying Lemma~\ref{lem:anti-concetration}\footnote{Although Lemma~\ref{lem:anti-concetration} concerns the number of $1$'s being chosen by processes, if the distribution is uniform then $0$ and $1$ have the same probability of occurring and, by the symmetry, the same lemma can be used to estimate the number of chosen $0$'s.} we get that with a constant probability $C_{1}$, the number of processes starting the next epoch $\cE_{2}$ with $b$ equal zero is at least $\frac{|\texttt{OP\_END}_{1}|}{2} + 3\sqrt{n}$. 
The constant $C_{1}$ is an appropriate constant in Lemma~\ref{lem:anti-concetration} corresponding to the fact that we have a $|\texttt{OP\_END}_{1}| \ge \frac{9}{10} n$ lower bound on the size of $\texttt{OP\_END}_{1}$ and we measure the deviation of size $\sqrt{n}$ from the mean.
Nevertheless, it follows that with probabiblity $C_{1}$ the number of processes having $b = 1$ is at most
$$|\texttt{OP\_END}_{1}| - \frac{|\texttt{OP\_END}_{1}|}{2} - 3\sqrt{t} \le \frac{|\texttt{OP\_END}_{1}|}{2} - 3\sqrt{n} \ .$$
In the next epoch $\cE_{2}$ at most $\sqrt{n}$ new processes become inoperative and thus the ratio of the number of $1$'s held by the operative processes to the total number of operative processes must remain lower than $\frac{1}{2}$ in the entire epoch $\cE_{2}$. It follows then, that all operative processes execute line~\ref{line:if-0} and assign $b$ to $0$ before the epoch ends. 

The second case is when at least one process executes line~\ref{line:if-1} in the epoch $\cE_{1}$. If there is at least one operative process assigning $1$ to its variable $b$ without sampling, and at most $\sqrt{n}$ processes become inoperative the epoch $\cE_{1}$, thus every other operative process assigns $1$ to the variable $b$ or it assigns a random value. 
Let $x$ be the number of processes that assign a random value in the epoch $\cE_{1}$.
We proceed with two subcases:
\begin{itemize}
    \item \textbf{Subcase 1:} Assume that $x \ge \frac{8}{10}|\texttt{OP\_END}_{1}|$. By Lemma~\ref{lem:anti-concetration}, with a constant probability $C_{2}$, more than $\frac{4}{10}|\texttt{OP\_END}_{1}| + 3\sqrt{n}$ from processes who execute line~\ref{line:if-random} assign $0$ to their value $b$. Thus, by the beginning of the epoch $\cE_{2}$ at most
    $$|\texttt{OP\_END}_{1}| - \frac{4}{10}|\texttt{OP\_END}_{1}| + 3\sqrt{n} = \frac{6}{10}|\texttt{OP\_END}_{1}| - 3\sqrt{n}$$
    processes have value $b$ set to $1$. In the epoch $\cE_{2}$ at most $\sqrt{n}$ processes become inoperative, therefore the ratio of $ones$ to the total number of operative counted in any process will be smaller than $\frac{6}{10}$. Thus, we arrived at the case when no process can execute line~\ref{line:if-1} in the epoch $\cE_{2}$ and the two consecutive epochs are good. Using the same argument as for the first case of this lemma, we get that with a constant probability all operative processes have value $b$ set to $0$ before the epoch $\cE_{3}$ ends.
    
    \item \textbf{Subcase 2:} Suppose that $x < \frac{8}{10}|\texttt{OP\_END}_{1}|$. By Lemma~\ref{lem:anti-concetration} there is a constant probability that at most $\frac{4}{10}|\texttt{OP\_END}_{1}| - 3\sqrt{n}$ processes set the value of $b$ to $0$. This is because at most $\frac{8}{10}|\texttt{OP\_END}_{1}|$ processes assign a random value and all the remaining ones assign $1$ to the variable $b$.  Therefore, at least
    $$|\texttt{OP\_END}_{1}| - \frac{4}{10}|\texttt{OP\_END}_{1}| + 3\sqrt{n} = \frac{6}{10}|\texttt{OP\_END}_{1}| + 3\sqrt{n}$$
    processes start the epoch $\cE_{2}$ with $b$ set to $1$. Again, in the epoch $\cE_{2}$ at most $\sqrt{n}$ processes can become inoperative. In such a case, Lemma~\ref{lem:operative-contribution} proves that operative processes at least $\frac{6}{10}|\texttt{OP\_END}_{1}| + 2\sqrt{n}$ counts of the value $1$ and thus all operative processes must execute line~\ref{line:if-1} in the epoch $\cE_{2}$. Therefore all operative processes assign value $1$ to their variable $b$.
\end{itemize}
\end{proof}

In the following lemma we show how the operative processes assure termination on the same value for non-faulty processes, even in the case of the unlikely event that some of the operative processes do not set the variable $\texttt{decided}$ to $true$.

\begin{lemma}\label{lem:deciding-1}
With probability $1$, all non-faulty processes decide and the decision is on the same value.
\end{lemma}

\begin{proof}
Let $O$ be the set of the processes that are operative after the last epoch of the algorithm \textsc{OptimalOmissionsConsensus}. We can further divide the set $O$ into two disjoint classes. The class $D \subseteq O$ of processes that have the variable $\texttt{decided}$ set to $true$ and the class $U \subseteq O$ consisting of these processes that have the variable $\texttt{decided}$ set to $false$. 

First, we prove that regardless of the actually partition of the operative processes into the classes $D$ and $U$ all non-faulty processes decide. Lemma~\ref{lem:good-proc-are-large} assures that $|O| \ge n -3t \ge \frac{9n}{10} \ge 2t$, thus either $D$ or $U$ have size greater than $t$. 

If it is the case that $|D| > t$, then there must be at least one non-faulty processes in $D$. The existence of this process assures that every non-faulty processes from the set $\cP \setminus O$ receives a message in the line~\ref{line:final-receiv-1} of the algorithm, and in consequence decides. The remaining non-faulty processes are all operative and they decide in line~\ref{line:if-decided} in the case they belong to the set $D$, or in line~\ref{line:spread-2} in the case they belong to $U$. The termination in the latter case follows from the correctness of the deterministic protocol from~\cite{DolevS83}, Theorem~4. 
For the case $|U| > t$ we have the following reasoning. The non-faulty operative processes can either belong to $D$ or $U$. Every non-faulty process belonging to $D$ decides in line~\ref{line:if-decided}. All non-faulty processes belonging to $U$ cannot decide in lines~\ref{line:good-spread}-\ref{line:if-decided} and they will execute line~\ref{line:spread-2}. From the termination of the deterministic protocol used in line~\ref{line:spread-2}, the non-faulty processes from $U$ will reach a decision and propagate the decision to all remaining processes that do not terminated yet. Since $|U| > t$, thus in the propagation takes part at least one non-faulty process which guarantees that every non-faulty processes will receive a decision in line~\ref{line:final-receiv-2}. This gives that eventually all non-faulty processes decide.    

Second, we show that all decision are on the same value.
To this end, we argue that if there exists a process $p$ that has the variable $\texttt{decided}_{p}$ set to $true$ when the last epoch ends then all operative processes have the value of the variable $b$ the same as the process $p$. 
Consider the epoch in which $p$, by executing line~\ref{line:if-1-r}, sets the variable $\texttt{decision}_{p}$ to $true$. It follows that the counts $\texttt{ones}_{p}$ and $\texttt{zeros}_{p}$ satisfy $\texttt{ones}_{p} > \frac{27}{30}(\texttt{ones}_{p} + \texttt{zeros}_{p})$ (or $\texttt{ones}_{p} < \frac{3}{30}(\texttt{ones}_{p} + \texttt{zeros}_{p})$, but since both cases are symmetric, we analyze only the first one). 
The number of operative process is always at least $n - 3t \ge \frac{9}{10}n$, by Lemma~\ref{lem:good-proc-are-large}. By Lemma~\ref{lem:operative-contribution}, it follows that the values of variables $\texttt{ones}$ and $\texttt{zeros}$ stored by different operative processes in the same epoch can differ by at most $4t$. Therefore, the ratio $\texttt{ones} / (\texttt{ones} + \texttt{zeros})$ calculated in any other operative process in this epoch is at least
$$\frac{\texttt{ones}_{p} - 4t}{\texttt{ones}_{p} + \texttt{zeros}_{p} + 4t}.$$
Since the process $p$ executes the line~\ref{line:if-1-r}, we get that the lower bound $\texttt{ones}_{p} \ge \frac{9}{10}(\texttt{ones}_{p} + \texttt{zeros}_{p})$. It follows that
$$\frac{\texttt{ones}_{p} - 4t}{\texttt{ones}_{p} + \texttt{zeros}_{p} + 4t} \ge \frac{\frac{9}{10}\left(\texttt{ones}_{p} + \texttt{zeros}_{p}\right) - 4t}{\texttt{ones}_{p} + \texttt{zeros}_{p} + 4t} \ge \frac{18}{30},$$
where the last inequality follows from a simple calculation, taking into account that $\texttt{ones}_{p} + \texttt{zeros}_{p} \ge \frac{9}{10}n$ and $t \le \frac{1}{30}n$. It yields, that the inequality in line \ref{line:if-1-r} holds for any other operative process in this epochs. Subsequently, this gives that any operative process sets the variable $b$ to $1$, by executing line~\ref{line:if-1}, if $p$ sets the variable $\texttt{decide}_{p}$ to $true$ in this epoch. 
By examining lines~\ref{line:good-spread}-\ref{line:fixing-und} we see that any process can decide only on a value that originates in a operative process. Since the consensus protocol used in line~\ref{line:spread-2} satisfies the validity condition, we have that its decision, if any, must also be equal to the decision of any operative process. This way there is only one value propagated as the decision in the system and every non-faulty process must decide on this value regardless of when it receives the value.

It remains to  a situation in which all operative processes have the variable $\texttt{decided}$ set to $false$. In this case, no process can decide before reaching line~\ref{line:spread-2}. The decision in any process is the outcome of the consensus protocol employed in line~\ref{line:spread-2} and by correctness of the employed algorithm, it must be the same across all process that decide.
\end{proof}

\noindent 
Finally, we give the main theorem.

\begin{theorem}
\label{thm:main-algo}
The algorithm \textsc{OptimalOmissionsConsensus} solves with probability~$1$ consensus against the adaptive adversary capable of controlling $t < \frac{n}{30}$ processes. With high probability the number of rounds used by the algorithm is $O\left(\frac{t}{\sqrt{n}} \log^{2} n\right)$ rounds and the number of communication bits is $O\left(n\left(t \log^{3}n + n\right)\right)$. With probability $1$ it uses $O\left(t\sqrt{n}\log^2{n}\right)$ random bits.
\end{theorem}

\begin{proof}
Let us first argue about the correctness of the algorithm.
The termination and agreement conditions the consequence of Lemma~\ref{lem:deciding-1}. The validity condition follows from the following observation. If all processes start an epoch with the same value of the variable $b$, either $0$ or $1$, then with probability $1$ all operative processes will have the same value of the variable $b$ (consistent with the initial bit) in all subsequent epochs. This is because in the algorithms~\textsc{GroupBitsAggregation} and $\textsc{GroupBitsSpreading}$ processes count only this values that where present in a process at the beginning of the epoch. Consequently, in such a case no process ever accesses a random source in line~\ref{line:if-random} of the algorithm \mbox{\textsc{OptimalOmissionsConsensus}}, as the ratio of $1$'s of operative processes to all values is always either $0$ or $1$. Since the decision of any non-faulty algorithm is always derived from a value of the variable $b$ of an operative process, c.f. lines~\ref{line:good-spread}-\ref{line:fixing-und}, thus the validity follows. 

We now give the analysis of the number of rounds, the number of bits and the number of random bits used by the algorithm. 
By Lemma~\ref{lem:good-proc-are-large}, there are at most $4t$ processes that might become inoperative during the run of the algorithm. On the other hand, any epoch that is not good has more than $\sqrt{n}$ processes that become inoperative in this epoch. Since the algorithm executes $\frac{t}{\sqrt{n}}\log n$ epochs, thus by the pigeonhole principle there is at least $\Omega(\log{n})$ disjoint sequences of three consecutive \textit{good} epochs. 

By Lemma~\ref{lem:three-epochs}, after each such triple, there is a $\Omega(1)$ probability that all operative processes have the same value of the variable $b$. Since, as argued above in this proof, if the operative processes start an epoch with the same value of the variable $b$ they start with same value of the variable $b$ in the following epoch, thus the probability that after all such good sequences of epochs, the operative processes do not have the same value in variable $b$ is at most $O(1)^{\log{n}} = \frac{1}{n^{C}}$, for some absolute constant $C > 0$. Equivalently, it holds that with probability $1 - \frac{1}{n^{C}}$ the operative processes have the same value in the variable $b$ in the end of the last epoch. Moreover, with this probability, all operative processes have also the variable $\texttt{decision}$ set to $true$. This holds because if a triple of good epochs ends strictly before the last epoch, then all values of the variable $b$ stored by the operative processes are the same, and in the next epoch all operative processes must set the variable $\texttt{decided}$ to $true$ when executing line~\ref{line:if-1-r}.  

Assuming that the operative processes store $true$ in their variable $\texttt{decided}$, then all operative processes execute line~\ref{line:good-spread} of the algorithm.
Since there is at least $n - 3t$ operative processes and at most $t  < \frac{1}{30}n$ faulty processes, in the next round each non-faulty process receives the value of at least one operative process and, in result, adopts the value as its decision. 
It follows that in line~\ref{line:if-decided} all non-faulty processes decide. The operative processes because they all have the variable $\texttt{decided}$ set to $true$; the non-faulty processes that are not operative decide because they received a message from an operative one. Since 
each epoch takes $O(\log{n})$ rounds and there is $\frac{t}{\sqrt{n}} \cdot \log{n}$ epochs, we have that all non-faulty processes terminate in $O\left(\frac{t}{\sqrt{n}} \log^2{n}\right)$ rounds whp. 

To upper bound the communication bit complexity, we first analyze how many bits are sent per epoch.  The algorithm \textsc{GroupBitsAggregation} uses $O(n\log^2{n})$ bits per group, as it is given in Lemma~\ref{lem:msg-aggr}. Since the algorithm is executed in parallel on $\ceil{\sqrt{n}}$ groups, this results in $O(n^{3/2} \log^{2}{n})$ bits per call to this algorithm by the main one.
The algorithm \textsc{GroupBitsSpreading} takes $O(\log{n})$ rounds of communication on a graph with the maximum degree $O(\log{n})$, according to Theorem~\ref{thm:random-graph-properties}. Per each link of the communication graph, processes send at most $O(\sqrt{n}\log{n})$ bits, amortized -- in a round, each process sends only the information that has not been passed yet and the total amount of information corresponding to the size of the array $\texttt{BitPacks}$, which contains $O(\sqrt{n}\log{n})$ bits. 
Since there are $O(\frac{t}{\sqrt{n}}\log{n})$ epochs, the communication bits' complexity incurred by processes during all epochs is $O(t \cdot n\log^{3}{n})$. The consensus protocol from~\cite{DolevS83} (Theorem 4.) incurs at most $O(n^2t)$ additional bits in its communication, but this can happen with probability at most $\frac{1}{n^{C}}$ for some absolute constant $C$. On the other hand, the number of bits used for informing the inoperative processes is $O(n^2)$ with probability $1$, as each operative process broadcasts its value at most once. Thus, the bound on the communication bit complexity follows.

For the upper bound on number of random bits, we observe that the use of those is predetermined by the number of epochs. Each operative process uses at most one random bit in an epoch, thus the total number of random bits used by all processes is bounded by $O\left(t\sqrt{n}\log^2{n}\right)$.
\end{proof}

\subsection{Comparison to the state-of-the art strategies for solving consensus in the model of crash failures} 
When the power of the adversary is limited to crashing processes, meaning that a process becomes completely disconnected from other processes at some point of the execution (which is a decision of the adversary), there exist optimal, or almost-optimal, algorithms with respect to time and bit complexity. 

As for the time efficient strategies, most of them rely on the time-optimal algorithm proposed by Bar-Joseph and Ben-Or~\cite{Bar-JosephB98}, who employed the idea of the random coin in the case where there is no clear majority of either preferred value: $0$ or $1$, among the processes. 
The use of the randomness in the case of crashes shares the same intuition with the approach of omissions. Because the standard bounds on the deviation from the mean on many i.i.d. random variables guarantee that, with constant probability, either the number of $0$'s exceeds the numbers of $1$'s by $\Theta(\sqrt{n})$ or vice-versa. 
In any case, the adversary is expected to crash $\Omega(\sqrt{n})$ processes to 
maintain decision uncertainty preserve the status quo in the system. 
Nevertheless, the important difference in the case of crash failures is that a crashed process stops to communicate with other processes starting from the round when the crash has occurred or the next round -- hence, every correct process can assume the same default value of the failed process with at most one round delay. 
In consequence, the adversary cannot prolong the strategy of crashes mitigating the deviations from the mean for more than $O(\sqrt{n})$ rounds as it would require controlling to many processes. 
This crucial property does not hold when omissions are allowed. A process, controlled by the adversary, may deliberately avoid communication with some chosen process (which would force the chosen process to assume a default input value of the faulty one), while at the same time -- informing some set of other processes about its actual value, which may be different from the default one. Even more maliciously, a faulty process may change the set of processes it communicates with from round to round which allows the adversary for even more flexibility. This motivates the partition of processes into classes that are downgrade monotonic rather than of unpredictable behavior. A contribution of our paper is to design such a partition, into operative / inoperative processes, and show that the random coin idea can be efficiently implemented on this partition.

The state-of-the-art bit efficient strategy was proposed by 
Hajiaghayi et al.~\cite{DBLP:conf/stoc/HajiaghayiKO22}. They designed a randomized consensus algorithm against crash failures, which uses $O\left(n^{3/2} \cdot \polylog(n)\right)$ bits whp and terminates in the almost-optimal time $O\left(t / \sqrt{n} \cdot \polylog(n)\right)$ (but has no provably almost-optimal communication bit complexity),
by exploiting certain ``locally compact'' properties of expander graphs of gradually growing degree, 
used for scheduling communication between processes. 
Unfortunately, the approach in~\cite{DBLP:conf/stoc/HajiaghayiKO22}
is not efficient against the omission failures controlled by the adaptive, full-power adversary. Similarly to~\cite{Bar-JosephB98}, \cite{DBLP:conf/stoc/HajiaghayiKO22} also relies on the fact that processes permanently stop after crashes, which allows them to amortize time or communication to the number of fail-stops, e.g., by doubling the number of contacted processes each time when too few responses are received.
This no longer works against omissions, because the adversary can control incoming/outgoing messages of the process that implements such doubling strategy, and enforce that the process inquires $\Theta(n)$ other processes before the adversary allows it to receive any messages. This way even a single omission-faulty process may contribute linearly to the communication complexity, while for crash failures such contribution could be amortized to a polylogarithm per crashed process. 

\section{A Lower Bound}
\label{sec:proof-lower}

We give a new lower bound, announced in Theorem~\ref{thm:lower-randomness-res}, showing that any consensus algorithm achieving $T$ round complexity with high probability, has to make at least $\Omega\left(\frac{t^{2}}{\log{n}}\right)$ calls to a random source, with probability at least $1-\frac{1}{\log{n}}$. The main probabilistic tool used in the analysis of the lower bound strategies, introduced later in this section, is an abstract one-round coin-flipping~game.

\noindent \textbf{The coin-flipping game, revisited.}
The sides of the coin-flipping game are: $k$ players, an adversary and a function $f$ that decides the outcome of the game
and is known to both the adversary and players. The game has the following organization. First, the players propose input values. The value of a player $p$ is drawn from an arbitrary distribution $X_p$, independently from other players.
Next, the adversary looks at the drawn values and has the power to \textit{hide} some subset of them (which we will also refer to as ``failing players'' or ``taking over players''). The hidden values are denoted by $\perp$. Finally, the outcome of the game is a binary value determined by the evaluation of the function 
$f : \{X_{1} \cup \perp\} \times \ldots \times  \{X_{k} \cup \perp\} \rightarrow \{0, 1\} $ in the point that corresponds to actions of the players and the adversary. We say that a sequence of players' values $y = (y_{1}, \ldots, y_{k}) \in X_{1} \times \ldots \times  X_{k}$ can be {\em biased towards a value $v \in \{0, 1\}$} if the adversary can change some players' values to $\perp$ obtaining a sequence $y' \in\{X_{1} \cup \perp\} \times \ldots \times  \{X_{k} \cup \perp\} $ such that $f(y') = v$. 

\noindent \textbf{Application of the coin-flipping game to the lower bound proof.}
The usefulness of this game for proving our lower bound is best visible if the processes are seen as players, the values $X_{p}$ are possible state transitions of $k$ (out of $n$) processes that do random calls, the ``hiding'' action of the adversary is to omit all links of the process in the round corresponding to the game,
and the binary outcome of the function is a predetermined, by the adversary, classification of executions (i.e., executions that are more likely to output $1$ vs those that are more likely to output $0$).

\noindent \textbf{The coin-flipping game -- technical part.}
In~\cite{Bar-JosephB98}, the authors show that with probability at least $1 - \frac{1}{n}$, an arbitrary game of $n$ players can be biased towards a particular outcome if the adversary can hide / fail $\Omega(\sqrt{n\log{n}})$ players. Below, by applying 
parameterized (by probability $\alpha$) Talagrand's concentration inequality, we generalize this result and show that even if only $k = o(n)$ players use randomness, the probability of biasing the game can be exponentially (in $k$) close to $1$. In contrast, in~\cite{Bar-JosephB98} this relation is linear and flat in the number of all players. Formally, in our proof we will rely on the Alon's and Spencer's formulation of Talagrand's inequality for convex bodies that originally appeared in~\cite{Talagrand95}, adjusted to the notation used in our paper.

\begin{theorem}[Theorem 7.6.1 in~\cite{alon2016probabilistic}]\label{thm:talagrand}
Let $\Omega$ be a probability space and let
\[
\Omega^{k} = \Omega \times \Omega \times \cdots \times \Omega
\]
be a product probability space. 
Then, for any $U \subseteq \Omega^{k}$ and for any $t \ge 0$, it holds that
\[
\Pr[U] \cdot \Pr\left[{U^c_t}\right] \le e^{-t^2/4}
\ ,
\]
where ${U^c_t}$ is the complement of $U_{t}$
defined as follows:
\[
U_t = \{ x \in \Omega^{k} ~:~ \rho(U,x) \le t \}
\]
and where $\rho(A, x)$ is the Talagrand's convex distance defined as
\[
\rho(U,x) =  \max_{h, \|h\|_2 \le 1} \rho_{h}(U, x) \ , \ \ \mbox{ where } \ \ \rho_{h}(U, x) = \min_{y \in U} \ \sum_{i~:~x_i \neq y_i} h_i
\]
for $h = (h_1,\ldots,h_k) \in \mathbb{R}^k$ and $x,y \in \Omega^{k}$.
\end{theorem}

\begin{lemma}
\label{cor:coin-game}
For any $\alpha \le \frac{1}{2}$, one can bias the single-round coin-flipping game towards one particular outcome $v \in \{ 0, 1\}$, with probability greater than $1-\alpha$, by hiding at most $8\sqrt{k\log(\alpha^{-1})}$ players' values. 
\end{lemma}

\begin{proof}
First observe that if $k=0$ then there is only one outcome of the game which happens with probability $1$. Therefore, the lemma is trivially proven. In the remainder, we assume that $k \ge 1$.

Consider the set $U^{0} \subseteq  X_{1} \times \ldots \times X_{k}$ which corresponds to those sequences that {\em cannot} be biased towards the outcome $0$ by replacing at most $8\sqrt{k\log(\alpha^{-1})}$ values with the default value $\perp$. 
If $\Prob(U^{0}) \le \alpha$ then the lemma is proven with $v = 0$. 

Assume that $\Prob(U^{0}) > \alpha$. Let $B(U^{0}, t)$ be the set of the sequences from $X_{1} \times \ldots \times X_{n}$ that differ with an element of $U^{0}$ on at most $t$ values. We apply Talagrand's concentration inequality (in the version of Theorem~\ref{thm:talagrand}) to prove that $\Prob\left(B\left(U^{0}, 8\sqrt{k\log(\alpha^{-1})}\right)\right) \ge 1-\alpha^{15}$. 
Observe that for vector $h$ such that $h_i=\frac{1}{\sqrt{k}}$, we get that $\|h\|=1$ and thus
\begin{equation}
\label{eq:Tal}
B\left(U^{0}, 8\sqrt{k\log(\alpha^{-1})}\right)
=
\{x:\rho_h(U^{0},x)\le 8\sqrt{\log(\alpha^{-1})}\}
\supseteq
\{ x: \rho(U^{0},x)\le 8\sqrt{\log(\alpha^{-1})}\}
\ ,
\end{equation}
where the first equality follows from the fact $h_{i} = \frac{1}{\sqrt{k}}$, in which case $\rho_{h}(U, x)$ is just a weighted Hamming distance, and the latter inclusion from the fact that for any $x$ it holds $\rho_{h}(U^{0},x) \le \rho(U^{0},x)$.
From Theorem~\ref{thm:talagrand} we get that
\[
\Prob\left(U^{0} \right) \cdot \Prob\left(\{ x: \rho(U^{0},x)\le 8\sqrt{\log(\alpha^{-1})}\}^c\right) 
\le
\exp\left(-(8\sqrt{\log(\alpha^{-1})})^2/4\right)
=
\alpha^{16}
\ .
\]
Since we assumed $\Prob\left(U^{0} \right) > \alpha$, then we can rewrite
\[
\Prob\left(\{ x: \rho(U^{0},x)\le 8\sqrt{\log(\alpha^{-1})}\}^c\right) 
\le \alpha^{16} \cdot \Prob\left(U^{0} \right)^{-1}
\le \alpha^{15}
\ ,
\]
which together with Property~(\ref{eq:Tal}) implies
\[
\Prob\left(B\left(U^{0}, 8\sqrt{k\log(\alpha^{-1})}\right)\right)
\ge
\Prob\left(\{ x: \rho(U^{0},x)\le 8\sqrt{\log(\alpha^{-1})}\}^c\right)
\ge
1-
\alpha^{15}
\ .
\]
To finalize the proof of the lemma, we observe that if a sequence $y$ belongs to $B\left(U^{0}, 8\sqrt{k\log(\alpha^{-1})}\right)$, it differs with an element $b \in U^{0}$ on at most $8\sqrt{k\log(\alpha^{-1})}$ positions. Let $b'$ be derived from $b$ by replacing those positions with $\perp$.  Since $b \in U^{0}$, therefore $f(b') = 1$ which in turn gives that $f(y') = 1$ where $y'$ is derived from $y$ by replacing the same positions with $\perp$. Therefore, $y \notin U^{1}$ as it can be biased towards $1$ by changing at most $8\sqrt{k\log(\alpha^{-1})}$ positions. 
We therefore have that 
$\Prob(U^{1}) \le 1 - \Prob\left(B\left(U^{0}, 8\sqrt{k\log(\alpha^{-1})}\right)\right) \le \alpha^{15}$, where the last inequality follows from the above use of Talagrand's inequality. Since $\alpha^{15} < \alpha$ for $\alpha < \frac{1}{2}$, thus the lemma is proven. 
\end{proof}

\noindent Consider now a set of $n$ processes and how their internal changes of states affect the state of the entire execution. We use Lemma~\ref{cor:coin-game}, with $\alpha=n^{-3}$, to limit the actions of an adversary when only a subset of processes choose to evoke randomness when changing their states between two communication rounds. In the context of 
an execution of a consensus algorithm, we have~the~following.

\begin{corollary}
\label{cor:gen-biasing}
Consider a single-round coin-flipping game on a set of $n$ processes from which only $0 \le k \le n$ rely on random choices when changing their internal state while all others use deterministic transitions. Then, by failing at most $8\sqrt{k\log^{3}{n}}$ processes, an adversary can bias the outcome of the game towards $0$ or $1$ \footnote{Note that the values $0$ and $1$ here are not necessarily tied to a consensus decision, but rather mark two different states of the whole system of processes.}
with probability at least $1 - \frac{1}{n^{3}}$. 
\end{corollary}

\noindent \textbf{Description of adversarial strategies.}
We say that a consensus algorithm is {\em $p$-strongly-correct}, for $p\in (0,1]$, if it satisfies all three conditions together (i.e., agreement, validity and termination) with probability at least $p$ against any adaptive adversary.\footnote{%
The adjective ``strongly'' is to distinguish from a weaker version in which each property must hold {\em separately} with probability at least $p$.} 
In the following, we restrict ourselves only to $1 - \frac{1}{n^{3/2}}$-strongly-correct algorithms. Recall also that 
each algorithm is
structured into rounds that are further split into two phases: a local computation phase (in which calls to a random source are included) and a communication phase. Without loss of generality, we can assume that a local computation phase begins with making a decision if a process is ready to decide. 
We keep track of an execution of a given
algorithm only until the first process decides. 
Consider an execution $\cE$ of an algorithm. 
For any round $i$ of the execution, we introduce the following notation -- let 
\begin{itemize}
    \item $\cH_i$ be the algorithm {\em history} of the whole execution, taken at the beginning of the local computation phase of round $i$ (aka the {\em state} of the algorithm),  
    \item
    $r_{i}$ be the number of processes that, based on their local subset of $\cH_{i}$, decide to use random bits in the local computation phase of the current round,
    \item $\cA_i$ be an adversarial strategy in round $i$ and the subsequent rounds, 
    under the history $\cH_i$ (we will drop sub-index $i$ if it is clear from the context),
    \item $\Prob(\cH_i,\cA)$ be the probability of reaching consensus on value $1$ when continuing the run of the algorithm with history $\cH_i$ under adversarial strategy 
    $\cA$.
\end{itemize}

\noindent We restrict the adversary to strategies in which it fails at most 
$16\sqrt{r_{i}\log{n}}+1=\Theta(\sqrt{r_i\log n}+1)$
processes in a round~$i$. 
Next, we introduce the classification of states of an execution, based on their potential valency -- the concept introduced in \cite{FischerLP85} for deterministic executions and in~\cite{Bar-JosephB98} for randomized ones.
The main difference in our classification, compared to~\cite{Bar-JosephB98}, is that we use the stronger Corollary~\ref{cor:gen-biasing} in the analysis and therefore we could refine the definition of bi-valent states to exclude scenarios that have too low probability of finishing either with decision value $1$ or $0$. This refinement, see below, is necessary to be able to analyze the amortized number of accesses to the random source.

\paragraph{Types of states (based on valency).}
We say that a state 
$\cH_i$ is: 

\begin{itemize}
    \item 
{\em null-valent} if for all adversarial strategies $\cA$ we have $\frac{1}{n\log n} - \frac{i}{n^{2}} \le \Prob(\cH_i,\cA)\le 1-\frac{1}{n\log n} + \frac{i}{n^{2}}$,
\item
{\em $1$-valent} if there is an adversarial strategy $\cA$ such that $\Prob(\cH_i,\cA)> 1-\frac{1}{n\log n} + \frac{i}{n^{2}}$ and for every other adversarial strategy $\cA'$: $\Prob(\cH_i,\cA') \ge \frac{1}{n\log n} - \frac{i}{n^{2}}$,
\item
{\em $0$-valent} if there is an adversarial strategy $\cA$ such that $\Prob(\cH_i,\cA) < \frac{1}{n\log n} - \frac{i}{n^2}$ and for every other adversarial strategy $\cA'$: $\Prob(\cH_i,\cA') \le 1- \frac{1}{n\log n} + \frac{i}{n^{2}}$,
\item
{\em bivalent} if there are adversarial strategies $\cA,\cA'$ such that $\Prob(\cH_i,\cA)> 1-\frac{1}{n\log n} + \frac{i}{n^2}$ and $\Prob(\cH_i,\cA') < \frac{1}{n\log n} - \frac{i}{n^{2}}$.
\end{itemize}

\noindent An execution that is $1$-valent or $0$-valent is also called {\em uni-valent}. We remark that the classes are disjoint and cover the whole space of an algorithm's states. 


\noindent \textbf{Overview of the lower bound proof.}
Initially, we show that there is an initial assignment of input bits to processes such that the algorithm starts in the state either bivalent or null-valent (c.f. Lemma~\ref{lem:initial-exe}). Consequently, we will prove that if the algorithm starts a round $i$ in an null-valent or bivalent state, then there is an adversarial strategy $\cA$ that keeps the algorithm state in one of these two classes for one more round, with high probability, or in the round 
every such
strategy has reached the limit of $t$ processes failed by the adversary (Lemmata~\ref{lem:null-valent}, \ref{lem:bi-valent}). 
In the latter case,
we show that because the strategy in each round failed at most $8\sqrt{r_{i}\log{n}}$ processes when $r_i\ge 1$ (or $1$ process if $r_i=0$), then a standard application of the Cauchy-Schwartz inequality yields that the desired amount of random calls has been used by the algorithm, with a constant probability. 
We now continue with the detail proof of Theorem~\ref{thm:lower-randomness-res}.


\begin{lemma}\label{lem:initial-exe}
For any synchronous consensus algorithm there exists an initial state, which, if the adversary can control one process, is null-valent or bivalent.
\end{lemma}

\begin{proof}
Assume to the contrary, that every initial state (i.e., every assignment of input values to processes) is either $1$-valent or $0$-valent. Certainly, the state $a_{1}$ where all inputs are $0$ has to be $0$-valent, while the state $b_{1}$ with only $1$'s as an input has to be $0$-valent. This holds, because the algorithm starting from all $0$'s has to solve consensus correctly (satisfying validity condition) with output $0$ with probability $1-\frac{1}{n^{3/2}}$ at least (and the opposite is true for all $1$'s). Consider now a chain of states $a_{1} = s_{1}, s_{2}, \ldots, s_{n} = b_{1}$ such that each two consecutive states are different only on one input value. Observe, that there must be a state $s_{j}$ in the sequence that is $0$-valent and the next state $s_{j + 1}$ is $1$-valent. Let $p$ be the process that gets different input in these two states. Consider the executions starting with $s_{j}$ but if the adversary fails the process $p$ and does not send its messages. If the execution becomes null-valent or bi-valent then the lemma is proven. If the execution remains $0$-valent, then the executions $s_{j+1}$ is bi-valent. It is $1$-valent by definition, however, by controlling (i.e. failing) the process $p$ and stopping it from sending messages it becomes $0$-valent, since there is no difference between this execution and the one started from $s_{j}$ with the adversary failing the same process. The same argument can be conducted if failing the process $p$ in the execution $s_{j}$ leads to $1$-valent execution. 
Since both these facts hold with high probability, their intersection also does and yields a contradiction.
Thus, the lemma is proven. 
\end{proof}

\begin{lemma} 
\label{lem:null-valent}
A state $\cH_{i}$ that is null-valent at the beginning of round $i < n$ can be extended to a null-valent state at the end of the round with probability greater than $1 - \frac{2}{n^{2}}$ by failing at most $16\sqrt{r_{i}\log n}$ processes.
\end{lemma}

\begin{proof}
Consider a one-round coin-flipping game applied to the state $\cH_{i}$ under the randomness used by the set of $r_{i}$ processes. 
Let us partition the possible outcomes (that is possible states in the subsequent round) of the game into two exhaustive classes:
\begin{description}
    \item[Class (a)] the states that are either bivalent or $1$-valent;
\item[Class (b)] the states that are either null-valent or $0$-valent.
\end{description}
By Corollary~\ref{cor:gen-biasing}, the adversary can bias the game by failing at most 
$8\sqrt{r_i\log{n}}$
processes towards either class (a) or (b), with probability at least $1-\frac{1}{n^{2}}$. If the adversary could bias towards class (a), obtaining a state $H_{i + 1}$ in the subsequent round
, then by the classification, there exists a further strategy $\cA_{i + 1}$ that guarantees
$$\Prob(\cH_{i + 1}, \cA_{i + 1}) > 1 - \frac{1}{n\log{n}} + \frac{i + 1}{n^{2}} \ .$$
It follows that at the beginning of round $i$ there was a strategy $\cA'_{i}$ that guaranteed
$$\Prob(H_{i}, \cA'_{i}) > \left(1- \frac{1}{n^{3}}\right)\left(1 - \frac{1}{n\log{n}} + \frac{i + 1}{n^{2}}\right)$$
$$=1 - \frac{1}{n\log{n}} + \frac{i + 1}{n^2} - \frac{1}{n^3} + \frac{1}{n^4\log{n}} - \frac{i + 1}{n^5} $$
$$\ge 1 - \frac{1}{n\log{n}} + \frac{i}{n^{2}} \ ,$$
where for the last inequality we used $i < n$, which contradicts the fact that $\cH_{i}$ is null-valent. Therefore, by failing at most 
$8\sqrt{r_{i} \log n}$ 
processes before the communication phase of round $i$, the adversary can steer to a \textit{temporary} state $H'_{i}$ satisfying $(b)$. Here, ``temporary'' refers to the fact that the state is measured between two phases of the round $i$.
Conditioned on this, we can view the remaining random choices of still-non-faulty processes as another one-round coin-flipping game with at most $r_{i}$ processes using randomness. 
We can further classify the outcomes (states) of this game into two classes:
\begin{description}
    \item[Class (a'):] the states that are $0$-valent;
    \item[Class (b'):] the states that are null-valent.
\end{description}
The adversary, by failing at most 
$8\sqrt{r_{i} \log n}$ 
other processes (thus altogether at most $16\sqrt{r_{i} \log n}$ processes), still between the phases of round $i$, can bias towards one of those new classes with probability at least $1 - \frac{1}{n^3}$, again by Corollary~\ref{cor:gen-biasing}. By analogous reasoning as above, we can exclude the class (a'), as biasing towards this class would mean that the state $\cH_{i}$ was $0$-valent. Therefore, it must be that the execution can be biased toward class (b') and henceforth the state $\cH_{i}$ can be extended to a null-valent state $\cH_{i + 1}$ in the next round with sufficiently high probability. 
\end{proof}


\noindent 
Next, we analyze the strategy for a state that is bivalent. 
In the proof, we use the fact that for any value of $r_i$, the upper bound $16\sqrt{r_{i}\log{n}}+1$ on the number of processes that the adversary can fail in a round is at least $1$.\footnote{Note that the first part of the upper bound formula, $16\sqrt{r_{i}\log{n}}$, which was sufficient for extending null-valent state as in Lemma~\ref{lem:null-valent}, could be $0$ for $r_i=0$, therefore we need $+1$ in the upper bound formula on the number of processes that could be failed by an adversary in a round, in order to give the adversary the power to fail at least one process in each round.}


\begin{lemma}
\label{lem:bi-valent}
Let $\cH_{i}$ be a bivalent state. By failing at most 
$16\sqrt{r_{i}\log{n}}+1$ 
processes per round, the adversary can extend the state for the next $i' > 1$ rounds, with probability at least $1 - \frac{i'}{n\log{n}}$, reaching a state $\cH_{i + i'}$  that is either bivalent or terminating. The latter case can happen only because failing the necessary processes in round $i + i' - 1$ would exceed the limit $t$ on the total number of failures.
\end{lemma}

\begin{proof}
Consider the state $\cH_{i}$ and let processes complete their local computation phase scheduled for this round. Let $\cH'_{i}$ be the random variable equal to the new state 
after
the local computation phase.
Since the communication phase involves no additional randomness, 
we can distinguish the following cases:
\begin{description}
    \item[Case A.] The randomness used by processes (in round $i$) led to the next state $\cH'_{i}$ that is either null-valent or bivalent --- then the adversary does nothing. 
\item[Case B.] The result of the random choices (in round $i$) of the processes led to an uni-valent state $\cH'_{i}$, say $1$-valent.
\end{description}
Nevertheless, since the state $\cH_{i}$ is bivalent, there exists a strategy $\cA_{i}$ that achieves $\Prob(\cH_{i}, \cA_{i}) < \frac{1}{n\log{n}} - \frac{i}{n^2} < \frac{1}{n\log{n}}$ (where the probability includes all possibilities for the value of the random variable~$\cH'_{i}$). 
Next, we shall prove that the fact that $\cH'_{i'}$ 
is $1$-valent, together with the existence of the strategy $\cA_{i}$, disallows termination of the algorithm with probability greater than $1 - \frac{1}{n}$ in the round $i + 1$, and in some cases -- even in later rounds, unless the adversary reached the limit of failed processes. 

If $0$ is to be decided in the next round, then the state must become either $0$-valent or bivalent at some point of implementing $\cA_{i}$. 
If it becomes bivalent, then the adversary stops implementing $\cA_{i}$. At this point, abandoning the strategy and letting the algorithm continue, leads to a bivalent state $\cH_{i+1}$ which proves the lemma.

The other case is when during implementing $\cA_{i}$, the state becomes $0$-valent. 
Assume that this happens after the failure of a process after all its messages are delivered. Then, the adversary does not fail this process and stops implementing the strategy. It follows that crashing or not crashing the process switches between $0$-valent and $1$-valent states. Therefore, the initial state of the round $i+1$ must be bivalent, as the adversary can fail the process at the very beginning of round $i + 1$ and switch between the two types. (Here we use a fact that the adversary has to be able to fail at least $1$ process per round, which is why we set up the upper bound $16\sqrt{r_{i} \log n}+1$ that is at least $1$ even for $r_i=0$.)
It can be also the case that the state becomes $0$-valent after failing a process together with failing to deliver some non-empty subset of its messages. Let $u$ be the recipient of the message that -- if delivered -- would keep the state $1$-valent, while if failed -- it would change the state to $0$-valent. 
The adversary stops the algorithm from delivering the message but does not take any further actions in round $i$. Now, failing the process $u$ at the beginning of round $i+1$ hides the information whether the last message has been delivered or not, which again was the only difference between $0$-valent and $1$-valent state -- thus the state must be bivalent, proving the lemma.

Finally, following the strategy $\cA_{i}$ by the adversary may not lead to the change of the valency of the state, resulting in a $1$-valent state $\cH_{i + 1}$ at the beginning of the next round. Anyway, if the algorithm decides at the very beginning of round $i+1$, the decision cannot be $0$. 
To achieve the $\frac{1}{n^{3/2}}$-strongly-correctness with the decision $0$, it must be $\Prob\left({\cH_{i + 1}, \cB}\right) < \frac{1}{n^{3/2}} $ for any strategy $\cB$. However, the state $\cH_{i+1}$ is $1$-valent, thus there exists a strategy $\cB'$ such that $\Prob\left({\cH_{i + 1}, \cB'}\right) > 1 - \frac{1}{n\log n} + \frac{i + 1}{n^2}$, which, assuming $i + 1 < n$, excludes decision on value $0$. 
On the other hand, implementing the strategy $\cA_{i}$ from the beginning of round $i$ guarantees that the probability of deciding $1$ is at most $\frac{1}{n\log{n}}$ (which follows from the choice of $\cA_{i}$). 
Therefore, the adversary can implement $\cA_{i}$ in round $i + 1$, then in round $i + 2$ and so on. Either it achieves a bivalent state in a round $i+i'$, for some $i' \ge 1$, or in some round it runs out of processes to fail. If the strategy $\cA_{i}$ was continued for $i'$ rounds, then the probability of early deciding, by the union bound, is at most 
$i' \cdot \frac{1}{n\log{n}}$.
Thus, the lemma is proved.
\end{proof}

\noindent We can now finalize the lower bound announced in Theorem~\ref{thm:lower-randomness-res}. For the sake of completeness of this section, we also decided to repeat the statement of the theorem and the proof.

\begin{theorem}[Theorem~\ref{thm:lower-randomness-res}]
\label{thm:lower-randomness}
For any $\left(1 - \frac{1}{n^{3/2}}\right)$-strongly-correct randomized algorithm solving consensus, let $T$ be the random variable denoting the number of rounds until the first process terminates, while $R$ be the random variable denoting the total number of calls to a random source. Then, there is an adversarial strategy that with probability at least $1 - \frac{1}{\log{n}}$ forces the algorithm to have
$$T \times \left( R + T \right) = \Omega \left(\frac{t^2}{\log{n}}\right) \ .$$
%
%
%
\end{theorem}

\begin{proof}
By Lemma~\ref{lem:initial-exe}, the adversary can assign input values such that the initial state is in either a bivalent or a null-valent state. Then, the adversary follows the strategy described in Lemmas~\ref{lem:null-valent} and ~\ref{lem:bi-valent}, depending whether the current state is null-valent or bivalent.
Specifically, if the state is null-valent, the adversary can extend the execution by one more round with probability at least $1 - \frac{1}{n^{2}}$, by Lemma~\ref{lem:null-valent}. If the state is bivalent, it can extend the state for some $i' > 1$ rounds 
with probability at least $1 - \frac{i'}{n\log{n}}$, such that the new state is again either bivalent or null-valent, or the execution terminates but then the number of failed processes in the previous round would exceed the adversary's limit $t$. If the algorithm decides to terminate, it must be in either a $0$-valent or $1$-valent state, since the algorithm is $\left(1 - \frac{1}{n^{3/2}}\right)$-strongly-correct. Therefore, the adversary can prolong the execution either for $n$ rounds or until it runs out of the processes to fail. 

Let $T$ be the round in which the execution terminated. If $T = n$, then the theorem
follows. 
Assume then that the adversary stopped implementing its strategy in round $T$ due the fact that in the preceding round it could not fail the desired number of processes. Since the adversary fails at most 
$16\sqrt{r_{i}\log{n}}+1$ 
processes in a round $i$ (c.f., Lemma~\ref{lem:null-valent} and~\ref{lem:bi-valent}), we obtain
\[t \le \sum_{i = 1}^{T - 1}\left(16\sqrt{r_{i}\log{n}}+1\right) 
\le 32 \sum_{i = 1}^{T - 1} \sqrt{(r_{i}+1)\log{n}} \ ,\]
which is equivalent to
\[ t^{2} \le 1024\left(\sum_{i = 1}^{T - 1}\sqrt{r_{i} \log{n}}\right)^{2} \ . \]
Applying the Cauchy-Schwarz inequality to the right hand side of the above, we get
\[ t^2 \le 
1024\left(\sum_{i = 1}^{T - 1}\sqrt{(r_{i} + 1)\log{n}}\right)^{2} 
\le 1024(T - 1)\left(\sum_{i = 1}^{T - 1} (r_{i} + 1)\log{n} \right) \ , 
\]
which, after proper rearranging, yields
\[ \frac{t^{2}}{1024\log{n}} \le (T - 1) \times (R + T) \]
and proves the theorem.
\end{proof}

\section{Trading Time for Randomness}\label{sec:tradeoff-upper}
Here, we present our algorithm that can trade random calls to a random source for time complexity. The building blocks of the algorithm are these three ideas: $(a)$ similarly to the optimal algorithm from Section~\ref{sec:main-algo} we use the partition of processes into operative and inoperative, rather than distinguishing between faulty and correct, ultimately we first achieve a common decision among the operative ones and then distribute it to the rest; 
$(b)$ to save on randomness, we partition the set of all processes $\cP$ into $x$ groups of size $\ceil{\frac{n}{x}}$ each. The high-level idea behind the partition is that achieving a consensus decision, by running the optimal algorithm from Section~\ref{sec:main-algo}, on a group of size $\ceil{\frac{n}{x}}$ processes uses less random calls than on the entire network, since the scalability is better than linear we can use this fact to exchange the number of random bits for the running time;  $(c)$ to carry over consensus decisions between groups we use round-robin approach. Each group has its time slot in which members of this group try to solve consensus. Then, the members of the group disseminate the consensus decision (if achieved) to other operative processes, who in their turn use this value as the input for the next consensus calls. Additionally, to achieve consensus conditions with probability $1$, processes employ a safety rule analogous to the one in \texttt{OptimalOmissionsConsensus}.
The pseudocode of the algorithm is given in Algorithm~\ref{alg:tradeoff}. The formal description is given below.

Let $SP_{1}, \ldots, SP_{x}$, for an integer $x \in [n]$, be a partition of the set of processes $\cP$ into $x$ groups (called also super-processes) of size $\ceil{\frac{n}{x}}$ each. Processes compute the partition, based on their identifiers, at the beginning of execution. Also, at the beginning of the execution, the processes compute a graph $G$ of the degree $\Theta(\log{n})$ in the same way as in the algorithm described in Section~\ref{sec:main-algo}, i.e., each process uses the result of Theorem~\ref{thm:random-graph-properties} to compute its set of neighbors in some predetermined graph $G$, which is guaranteed to exists by the theorem.

\begin{algorithm}
\SetAlFnt{\tiny}
\SetAlgoLined
\SetKwInput{Input}{input}
\Input{$\cP$, $p$, $b_p$}
$SP_{1}, \ldots, SP_{x} \leftarrow $ a partition of $\cP$ into 
$x$ disjoint super-processes of size $\frac{n}{x}$ each\;
$\texttt{operative}_{p} \leftarrow true$\;
$V_{p} \leftarrow$ $V_{p} \leftarrow$ a set of neighbors of $p$ in a predetermined graph $G$ guaranteed by Theorem~\ref{thm:random-graph-properties}\;
\For{$i \leftarrow 1$,  $i \le x$}
{
\lIf{$p \in SP_{i}$}
{
$\texttt{b}_{p}, 
\leftarrow \textsc{OptimalOmissionsConsensus}(SP_{i}, p, b_{p})$
}
\lElse{wait for a fixed number of $\Theta\left(\sqrt{\frac{n}{x}}\log^{2}\left(\sqrt{\frac{n}{x}}\right)\right) $ rounds}

\lIf{$p \in SP_{i}$}
{$\texttt{consensus\_decision}_{p} \leftarrow b_{p}$
}
\lElse{$\texttt{consensus\_decision}_{p} \leftarrow \perp$}
\BlankLine
\For{$2\log{n}$ rounds}
{
    \lIf{$\texttt{operative}_{p} = false$}{stay idle until line~\ref{line:final-receiv-param}\label{line:idle-param}}
    \textbf{for} $q \in V_{p}$ : \newline
    \hspace*{2.5mm} \textbf{send} $\texttt{consensus\_decision}_{p}$ if $q$ not disregarded before\newline
    \hspace*{2.5mm} \textbf{receive} $\texttt{consensus\_decision}_{q}$ from $q$ and update:\newline
    \hspace*{5.5mm} $\texttt{consensus\_decision}_{p} \leftarrow \texttt{consensus\_decision}_{p} \cup \texttt{consensus\_decision}_{q}$;\newline
    \hspace*{2.mm} \textbf{if} $q$ has not sent a message, disregard sending to $q$ in any future round;\newline
    \lIf{number of received messages is less than $\Delta / 3$}{\newline \hspace*{2.5mm} $\texttt{operative}_{p} \leftarrow false$}   
}
$b_{p} \leftarrow \texttt{consensus\_decision}_{p}$\;\label{line:param-acq-b}
}
\tcc{implementation of a safety rule}
$\texttt{decide} \leftarrow false$ \label{line:sft-starts-param}\;
\If{$\texttt{operative}_{p} = true$}{
\textbf{send} $b_{p}$ to all processes in $\cP$\label{line:operative-dis-param}\;
\textbf{receive} bits sent in the previous round; let variables $\texttt{ones}_{p}, \texttt{zeros}_{p}$ denote the number of $1$'s and $0$'s received\label{line:count}\;
\lIf{$\texttt{ones}_{p} > \frac{18}{30}(\texttt{ones}_{p} + \texttt{zeros}_{p})$}{$b_p \leftarrow 1$}\label{line:if-1-param} 
\lElseIf{$\texttt{ones}_{p} < \frac{15}{30} (\texttt{ones}_{p} + \texttt{zeros}_{p})$}{$b_p \leftarrow 0$}\label{line:if-0-param}
\lIf{$\texttt{ones}_{p} > \frac{27}{30}(\texttt{ones}_{p} + \texttt{zeros}_{p})$}{$\texttt{decided}_{p} \leftarrow true$}\label{line:if-1-r-param} 
\lElseIf{$\texttt{ones}_{p} < \frac{3}{30} (\texttt{ones}_{p} + \texttt{zeros}_{p})$}{$\texttt{decided}_{p} \leftarrow true$}\label{line:if-0-r-param}
}
\BlankLine
\lIf{$\texttt{operative}_{p} = true$ and $\texttt{decided}_{p} = true$}{\textbf{send} $b_{p}$ to all processes in $\cP$\label{line:good-spread-param}}
\lElseIf{any message $b_{q}$ received from some process $q$}{
$b_{p} \leftarrow b_{q}$\label{line:final-receiv-param}}
\tcc{in the above, $q$ can be chosen arbitrarily from the received messages}
\BlankLine
\lIf{$\texttt{decided}_{p} = true$ or $(\texttt{operative}_{p} = false$ and $p$ received a message in the previous round$)$}{\textbf{decide} $b_{p}$\label{line:if-decided-param}}
\Else{
\lIf{$\texttt{operative}_{p} = true$}{$p$ takes part in any deterministic synchronous consensus algorithm (e.g.,~\cite{DolevR85}); if $p$ reaches agreement in that protocol, it broadcast the decision to all processes in $\cP$ and it decides on the algorithm's decision\label{line:spread-param}}
\lElse{$p$ remains idle until a decision is sent to it; upon receiving a decision it decided on this value\label{line:final-receiv-para-2}}
}\label{line:fixing-und-param}
\caption{\textsc{ParamOmissions}}\label{alg:tradeoff}
\end{algorithm}

Then the round-robin stage of computation begins. Throughout this stage, processes store their candidate decision value in a variable $b$, which is initialized to their input value. Also, initially, each process starts the round-robin stage with the operative status set to $true$. This status is maintained based on the number of received messages (messages are send only via edges of the graph $G$).

The round-robin stage of computation is structured into $x$ phases. In phase $i$, for $1 \le i \le x$, operative processes belonging to the super-process $SP_{i}$ run a variation of the algorithm \textsc{OptimalOmissionsConsensus} from Section~\ref{sec:main-algo} on their values $b$. The difference is that the algorithm \textsc{OptimalOmissionsConsensus} is run for a fixed $\Theta\left(\sqrt{x}\log^{2} x\right)$ number of rounds, i.e. the \textsc{OptimalOmissionsConsensus} algorithm is terminated at line~\ref{line:if-decided}.  Although it might not always achieve a consensus decision, a simple observation from the proof of Theorem~\ref{thm:main-algo} is that it can happen with a polynomially small probability. On the other hand, we have a control over the number of rounds used for a single run. Processes of other super-processes stay idle for the fixed number of rounds. This is where the "trading" part of the algorithm takes place. The randomness needed to calculate the consensus value on the members of the super-processes is $\Theta\left(\left(\frac{n}{x}\right)^{3/2}\right)$, which scales better than linearly and therefore $x$ separate executions can save random bits compared to the baseline algorithm \textsc{OptimalOmissionsConsensus} from Section~\ref{sec:main-algo} (which is a single execution on the set of $n$ processes).

The outcome of the truncated at line~\ref{line:if-decided} execution of the \textsc{OptimalOmissionsConsensus} algorithm to a process can be two-fold. Either the process got a decision or it has received no messages with a decision when running the \textsc{OptimalOmissionsConsensus} internally (c.f. lines~\ref{line:good-spread}-\ref{line:if-decided} in Algorithm~\ref{alg:opt-omissions}). In the first case, the process substitutes its candidate value $b$ with the consensus value, in the second case, it substitutes its value $b$ with a default, null-symbol $\perp$. In the last part of the phase, operative processes of the super-process $SP_{i}$ floods the consensus value (if any) calculated earlier along the graph $G$. 
The flooding proceeds in $\Theta(\log{n})$ rounds. Each operative process sends a message to its neighbors in $G$. The message is either empty or contains the consensus decision from the beginning of the phase if the process has already learned this value (operative processes of $SP_{i}$ know the value from the beginning of the flooding). 
When receiving messages from neighbors, every process updates the current consensus decision upon receiving, and also keeps track of processes that have not sent any message to it. 
First, it excludes such processes from further flooding communication in this and any subsequent round-robin phase of the execution. 
Second, if the number of processes that sent a message drops below $\Delta / 3$, the process becomes inoperative and stops its participation in \textit{any} subsequent actions of the algorithm. 

The rationale behind the round-robin stage is that processes aim for a phase in which the algorithm \textsc{OptimalOmissionsConsensus} is executed on a super-process whose large fraction of processes is operative. In case like this, with high probability, all processes that remain operative till the end of the execution of \textsc{OptimalOmissionsConsensus} have the same consensus value (and this set is non-empty, assuming that the adversary can crash a small fraction of processes). Then, the flooding procedure on the random graph guarantees that the consensus value is propagated to all processes that are operative. At this point, there is only one value in the system, so it can be expected that this value will 
remain the consensus decision regardless of the adversarial strategy. Since in the above scenario there is a $O(1/n)$ chance that the operative processes fail to reach consensus, after the round-robin stage ends, the algorithm deploys a safety rule, cf.~lines~\ref{line:sft-starts-param}-\ref{line:fixing-und-param}. Specifically, the operative processes count the number of $1$'s and $0$'s among themselves and in the rare case that both these numbers are close to each other, they execute the deterministic protocol from~\cite{DolevR85}.

\noindent Below we present the theorem summarizing the \textsc{ParamOmissions} algorithm.

\begin{theorem}\label{thm:alg-tradeoff}
The algorithm $\textsc{ParamOmissionsConsensus}$ solves consensus (with probability $1$) against the adaptive adversary that can control $t < \frac{n}{60}$ processes. 
For any integer parameter $x \in [n]$, it 
achieves the following bounds, with high probability: it terminates in $\logO\left(\sqrt{xn}\right)$ rounds, uses $\logO \left(n\sqrt{\frac{n}{x}}\right)$ random bits and $ \logO\left(n\cdot \frac{n}{x} + xn+ n^2\right)$ bits of communication.
\end{theorem}

\noindent 
Note, that if we denote by $R$ the number of random  bits used by the algorithm, then by rearranging the remaining terms in the above theorem accordingly, we get the formulation of Theorem~\ref{thm:trade-off-res}.

\noindent \textbf{Analysis.}
Processes can become inoperative only by receiving less than $\Delta / 3$ messages from their neighbors in the graph $G$. This is a milder requirement compared to the one used in the algorithm \textsc{OptimalOmissionsConsensus}, where not only the process has to receive $\log{n}$ messages in the graph $G$, but it also has to remain operative during other sub-protocols. Also, we allow at most $\frac{n}{40}$ failures which are twice as small as in the proof of Theorem~\ref{thm:main-algo}, thus the following lemma is a consequence of Lemma~\ref{lem:good-proc-are-large} proven in Section~\ref{subsec:analysis-main}.
\begin{lemma}\label{lem:operative-param}
With high probability, at least $n - 3t \ge \frac{57n}{60}$ processes stay operative throughout the execution of the algorithm.
\end{lemma}
We say that a super-process $SP_{i}$ is \textit{reliable} if at least $\frac{29}{30}\frac{n}{x}$ processes belonging to it remain \textit{non-faulty} to the end of the execution and at least one of these process is \textit{operative} throughout the execution.
\begin{lemma}\label{lem:reliable-sp}
In any execution, there is at least one reliable super-process with high probability.
\end{lemma}
\begin{proof}
Since there are at most $\frac{n}{60}$ faulty processes, thus by the pigeonhole principle there are at least $\frac{x}{2}$ super-processes that contain at most $\frac{1}{20}\frac{n}{x}$ faulty processes. Consequently, at least $\frac{1}{2}\frac{29}{30} \cdot n$ processes belonging to these super-processes remain non-faulty throughout the entire execution. By Lemma~\ref{lem:operative-param}, there is another set of at least $n - 3t \ge \frac{57n}{60}$ processes that are always operative. Since $\frac{57}{60} + \frac{1}{2}\frac{29}{30} > 1$, thus these two sets have a non-empty intersection, proving that there is a super-processes that has at most $\frac{1}{30}\frac{n}{x}$ faulty members and among the non-faulty members at least one is operative to the end of the algorithm.
\end{proof}

Next, we recall a property of the flooding algorithm (e.g. Lemma~\ref{lem:operative-contribution}) used in \textsc{ParamOmissions} that was first used to prove Theorem~\ref{thm:main-algo}. 
Note, that that the flooding is designed in the same way as the \textsc{GroupBitsSpreading} algorithm used in \textsc{OptimalOmissionsConsensus} (cf. Algorithm~\ref{alg:bits-spread}), only the content of relayed messages is different. Therefore, by repeating the reasoning used in Lemma~\ref{lem:spreading-reaching}, we obtain the following.
\begin{lemma}\label{lem:spreading-param}
For any two processes, $p, q$ that are operative at the end of a phase $i$, the consensus decision of $p$ (if different than $\perp$) has been relayed to $q$ in the flooding part of the phase $i$ and vice-versa, the consensus decision of $q$ (if different than $\perp$) has been relayed to $p$.
\end{lemma}
Finally, we observe that the same arguments used for proving the correctness of the stopping rule in the algorithm~\textsc{OptimalOmissionsConsensus} apply to the stopping rule used in line~\ref{line:sft-starts-param}-~\ref{line:fixing-und-param} of the \textsc{ParamOmissions} algorithm. Thus we get the following.

\begin{lemma}\label{lem:deciding-1-param}
With probability $1$, all non-faulty processes decide and the decision is on the same value.
\end{lemma}

Next, we proceed to the proof of the main theorem of the section.
\begin{proof}[Proof of Theorem~\ref{thm:alg-tradeoff}]
The 
agreement and
termination follow from Lemma~\ref{lem:deciding-1-param}. The validity follows from the fact that the truncated version of \textsc{OptimalOmissionsConsensus} also satisfies validity, cf. the proof of Theorem~\ref{thm:main-algo}. 

Next, we analyze the complexity measures. Denote $SP_{i}$ the reliable super-process which, by Lemma~\ref{lem:reliable-sp}, exists with high probability. Since at least $\frac{29}{30}\frac{n}{x}$ processes belonging to $SP_{i}$ remain non-faulty, we can apply Theorem~\ref{thm:main-algo} and conclude that the truncated execution of the algorithm \textsc{OptimalOmissionsConsensus} on the group $SP_{i}$ in the appropriate phase of the round-robin scheme results in all non-faulty processes of $SP_{i}$ having the same current consensus decision $b$ with high probability. 
By the reliability of the super-process, among these non-faulty processes there exists at least one operative process. Thus, in the flooding part of the round-robin phase $i$, all other operative processes acquire the output of the \textsc{OptimalOmissionsConsensus} from the beginning of the phase, by Lemma~\ref{lem:spreading-param}. Therefore, after the flooding ends, all operative processes have the same value $b_{p}$, see line~\ref{line:param-acq-b}. 
Since in any phase, only operative processes execute the \textsc{OptimalOmissionsConsensus} algorithm (see line~\ref{line:idle-param}), each subsequent run of the algorithm has the property that all parties begin with the same input bit. 

We observe that in case when all processes start with the same input bit, processes that return a value at the end of the algorithm \textsc{OptimalOmissionsConsensus} can only return the value being the input bit. This follows from analyzing lines~\ref{line:if-1}-\ref{line:main-for-end} in the algorithm \textsc{OptimalOmissionsConsensus} (Algorithm~\ref{alg:opt-omissions}). If there is only one input value in the system, no process will access the random coin and only the input value will be manipulated to the end of the algorithm. It follows that after the round-robin phase performed on the reliable super-process $SP_{i}$, all operative processes have the same value in the variable $b_{p}$. In this case, in line~\ref{line:count} all operative processes receive only one bit and thus they always set the variable $\text{decided}$ to $true$ when executing conditional statements in lines~\ref{line:if-1-r-param}-\ref{line:if-0-r-param}. 
In consequence, all operative processes spread the same decision  in line~\ref{line:good-spread-param}. Since the number of operative processes is at least $\frac{57n}{60} > \frac{n}{60} = t$, every non-faulty process eventually receives a decision and decides on it in line~\ref{line:if-decided-param}. Observe also that in this case no operative process has its variable $\texttt{decided}$ set to false and thus the deterministic protocol in line~\ref{line:spread-param} will not be executed. 

The derivation of the complexity measures comes directly from conditioning on the event that all operative processes end having the variable $\texttt{decided}$ set to $true$, which from now on we assume.

As for the running time, every phase lasts $\logO \left(\sqrt{\frac{n}{x}} \right)$ rounds, because it consists of a single run of the truncated algorithm \textsc{OptimalOmissionsConsensus}, which time complexity can be easily derived from Theorem~\ref{thm:main-algo}, and of $2\log{n}$ rounds of flooding. There are $x$ phases and two additional rounds after the round-robin phases end, thus the bound on the time complexity follows with high probability.

We bound the communication bit complexity analogously. Each phase uses $O\left(\left(\frac{n}{x} \right)^2 \log^4 \left(\frac{n}{x} \right) + n \log^2{n} \right)$ communication bits. The first term corresponds to the execution of the algorithm \textsc{OptimalOmissionsConsensus} on a set of size $\frac{n}{x}$ (c.f., Theorem~\ref{thm:main-algo}), the latter corresponds to $4\log{n}$ rounds of flooding. Observe that in the flooding, in each round every process sends at most $\Theta(\log{n})$ messages of size $O(1)$. Multiplying the phase bit complexity by the number of phases and adding $O(n^2)$ term, corresponding to the last two rounds, explains the bound on the bit complexity of the whole algorithm and that it holds whp.

The random bits are spent only in the runs of \textsc{OptimalOmissionsConsensus}. There are at most $x$ independent runs of the algorithm, each on a set of processes of size $\frac{n}{x}$. Applying Theorem~\ref{thm:main-algo} yields that the total number of used random bits is $O\left(x \cdot \left(\frac{n}{x}\right)^{3/2} \log^2\left(\frac{n}{x}\right)\right) = \logO \left(n\sqrt{\frac{n}{x}}\right)$, which completes the proof of the theorem.
\end{proof}

\end{document}